\documentclass[12pt]{article}

\usepackage[utf8]{inputenc}

\usepackage{amsmath,amsfonts,amssymb,amsthm}
\usepackage{xspace}
\usepackage{enumerate}
\usepackage{url}

\usepackage{geometry}
\theoremstyle{plain}
\newtheorem{theorem}{Theorem}[section]
\newtheorem{lemma}[theorem]{Lemma}
\newtheorem{proposition}[theorem]{Proposition}

\newtheorem{corollary}[theorem]{Corollary}
\theoremstyle{definition} 

\newtheorem{remark}[theorem]{Remark}

\newcounter{claim}
\renewcommand{\theclaim}{\Alph{claim}}
\newenvironment{claim}{\refstepcounter{claim}%
\par\medskip\par\noindent{\it Claim~\theclaim.~}~\rm}{\par\vspace{-2mm}\par}
{$\,\triangleleft$\par\medskip\par}

\newcommand{\Case}[2]{\bigskip\par\noindent{\it Case #1:\/ #2}\par\medskip\noindent}

\makeatletter
\def\@gifnextchar#1#2#3{\let\@tempe#1\def\@tempa{#2}\def\@tempb{#3}%
  \futurelet\@tempc\@gifnch}

\def\@gifnch{\ifx\@tempc\@sptoken\let\@tempd\@tempb%
  \else\ifx\@tempc\@tempe\let\@tempd\@tempa\else\let\@tempd\@tempb\fi\fi\@tempd}

\def\SK@set#1{\left\{#1\right\}}
\def\SK@@set#1#2{\{#1\,:\,
    \begin{array}{@{}l@{}}#2\end{array}
\}}
\def\SK@mset#1{\left\{\!\!\left\{#1\right\}\!\!\right\}}
\def\SK@@mset#1#2{\{\!\!\{#1\,:\,
    \begin{array}{@{}l@{}}#2\end{array}
\}\!\!\}}
\def\BIG@set#1{\Big\{#1\Big\}}
\def\BIG@@set#1#2{\Big\{#1\:\Big|\:
    \begin{array}{@{}l@{}}#2\end{array}
\Big\}}
\newcommand{\Set}[1]{\@gifnextchar\bgroup{\SK@@set{#1}}{\SK@set{#1}}}
\newcommand{\Mset}[1]{\@gifnextchar\bgroup{\SK@@mset{#1}}{\SK@mset{#1}}}
\newcommand{\Bigset}[1]{\@gifnextchar\bgroup{\BIG@@set{#1}}{\BIG@set{#1}}}
\makeatother

\newcommand{\refeq}[1]{(\ref{eq:#1})}

\newcommand{\of}[1]{\left( #1 \right)}

\newcommand{\function}[2]{:#1 \rightarrow #2}

\DeclareMathOperator{\sym}{Sym}
\DeclareMathOperator{\aut}{Aut}

\DeclareMathOperator{\cay}{Cay}

\DeclareMathOperator{\rank}{rk}
\DeclareMathOperator{\ann}{Ann}
\DeclareMathOperator{\wele}{WL}
\DeclareMathOperator{\orb}{Orb}

\newcommand{\wll}[1]{\wele_2(#1)}
\newcommand{\orbb}[1]{\orb_2(#1)}

\newcommand{\WL}[1]{\ensuremath{#1\text{-}\mathrm{WL}}\xspace}
\newcommand{\wl}{\WL 2}
\newcommand{\kwl}{\WL k}

\newcommand{\bQ}{\mathbb{Q}}

\newcommand{\bZ}{\mathbb{Z}}
\newcommand{\bC}{\mathbb{C}}
\newcommand{\bU}{\mathbb{U}}
\newcommand{\cC}{\mathcal{C}}

\newcommand{\cS}{\mathcal{S}}

\newcommand{\reals}{\mathbb{R}}

\newcommand{\rgraph}{\mathbf{G}_n}

\newcommand{\prob}[1]{\mathsf{P}[ #1 ]}
\newcommand{\probla}[1]{\mathsf{P}_{\mathfrak{l}}[ #1 ]}
\newcommand{\probun}[1]{\mathsf{P}_{\mathfrak{u}}[ #1 ]}

\newcommand{\un}[1]{{#1}^{\mathfrak{u}}}
\newcommand{\la}[1]{{#1}^{\mathfrak{l}}}
\newcommand{\Q}{\mathcal{Q}}
\newcommand{\R}{\mathcal{R}}
\newcommand{\asym}{\ensuremath{\mathcal{A}}\xspace}
\newcommand{\firm}{\ensuremath{\mathcal{F}}\xspace}
\newcommand{\allc}{\ensuremath{\mathcal{C}}\xspace}

\newcommand{\fourier}{\mathcal{F}}

\title{Canonization of a random circulant graph by counting walks\footnote{A preliminary
    version of this paper was presented at WALCOM'24, the 18th International Conference
                  and Workshops on Algorithms and Computation~\cite{walcom}.}}

\author{Oleg Verbitsky\thanks{Institut f\"{u}r Informatik, Humboldt-Universit\"{a}t zu Berlin, Germany.
Supported by DFG grant KO 1053/8--2. On leave from the IAPMM, Lviv, Ukraine.}
  \and
Maksim Zhukovskii\thanks{School of Computer Science, University of Sheffield, UK.}}

\date{}

\begin{document} 

\maketitle

\begin{abstract}
  It is well known that almost all graphs are canonizable by a simple combinatorial routine known as color refinement, also referred to as the 1-dimensional Weisfeiler-Leman algorithm.
  With high probability, this method assigns a unique label to each vertex of a random input graph and, hence, it is applicable only to asymmetric graphs. The strength of combinatorial refinement techniques becomes a subtle issue if the input graphs are highly symmetric.
We prove that the combination of color refinement and vertex individualization yields a canonical labeling for almost all circulant digraphs (i.e., Cayley digraphs of a cyclic group).
This result provides first evidence of good average-case performance of combinatorial refinement within the class of vertex-transitive graphs.
Remarkably, we do not even need the full power of the color refinement algorithm. We show that the canonical label of a vertex $v$ can be obtained just by counting walks of each length from $v$ to an individualized vertex. Our analysis also implies that almost all circulant graphs are compact in the sense of Tinhofer, that is, their polytops of fractional automorphisms are integral.
Finally, we show that a canonical Cayley representation can be constructed for almost all circulant graphs by the more powerful 2-dimensional Weisfeiler-Leman algorithm. 
\end{abstract}

\section{Introduction}\label{s:intro}

\subsection{Combinatorial refinement and canonization of random graphs}

As is well known, the graph isomorphism problem is very efficiently solvable
in the average case by a simple combinatorial method known as \emph{color refinement}
(also \emph{degree refinement} or \emph{naive vertex classification}).
When a random graph $\rgraph$ on $n$ vertices is taken as an input,
this algorithm produces a canonical labeling of all vertices in $\rgraph$
by coloring them initially by their degrees and then by refining the color classes as follows:
Two equally colored vertices $u$ and $v$ get new distinct colors if one of the initial colors
occurs in the neighborhoods of $u$ and $v$ with different multiplicity.
In this way, every vertex gets a unique color with probability $1-O(n^{-1/7})$
(Babai, Erd\H{o}s, and Selkow \cite{BabaiES80}).
Thus, the method produces a canonical labeling for almost all graphs on a fixed set of $n$ vertices.

This approach is not applicable to regular graphs, even with
many refinement rounds, because if all vertices have the same degree, then
the refinement step makes obviously no further vertex separation.
Weisfeiler and Leman \cite{WLe68} came up with a more powerful refinement algorithm which colors
pairs of vertices instead of single vertices (see Section \ref{ss:alg} for a formal description). The idea can be lifted
to $k$-tuples of vertices, for each integer parameter $k$, and the general
approach is referred to as the $k$-dimensional Weisfeiler-Leman algorithm,
abbreviated as \kwl. Thus, \wl is the original algorithm in \cite{WLe68},
and \WL1 corresponds to color refinement.
Remarkably, \wl is powerful enough to produce a canonical labeling for almost all regular graphs
of a given vertex degree (Bollob{\'a}s \cite{Bollobas82}; see also~\cite{Kucera87}).

Further restriction of regular input graphs to vertex-transitive graphs is challenging for
combinatorial refinement because no vertex classification is at all possible in this case.
Indeed, \wl assigns the same color to any two vertices $u$ and $v$ (or, more precisely, to the pairs $(u,u)$ and $(v,v)$) because $u$ is mapped
to $v$ by an automorphism of the graph. The same holds for any dimension~$k$.

Moreover, the graph isomorphism problem for vertex-transitive graphs is provably unsolvable
by \kwl for any fixed dimension $k$.
Indeed, the Cai-F\"{u}rer-Immerman (CFI) construction \cite{CaiFI92}
of non-isomorphic graphs $X_k$ and $Y_k$ indistinguishable by \kwl can be modified so that
these graphs become vertex-transitive~\cite{FuhlbruckKPV21}.
A natural way to enhance combinatorial refinement is to combine
it with \emph{vertex individualization}---that is, assigning a special color to one vertex in the graph \cite{McKayP14}. While this algorithmic approach
proves advantageous in many contexts (see, e.g.,~\cite{BabaiCSTW13}), it nevertheless fails
to overcome the obstacle posed by the CFI graphs.
To see this, note that \WL{(k+k')} is more powerful than any combination of \kwl with prior individualization
of $k'$ vertices. As a consequence, for any pair of arbitrarily large integers $k$ and $k'$, isomorphism of vertex-transitive graphs cannot be
solved by \kwl even under individualization of $k'$ vertices.

Motivated by the question of whether or not these basic obstacles persist in the average case setting,
we focus in this paper on Cayley graphs and, more specifically, on circulant graphs,
that is, Cayley graphs of a cyclic group. While the canonization problem for this class
of graphs is known to be solvable in polynomial time \cite{EvdokimovP04} by advanced algebraic methods,
it is an open question whether this can be done by using \kwl for some fixed dimension $k$; see~\cite{Ponomarenko23,WuPonomarenko24}.
This poses an ongoing challenge for the combinatorial refinement method, especially because
the research on isomorphism of circulant graphs has a long history with many deep results
(see \cite{Babai77,EvdokimovP04,MuzychukKP99,Muzychuk04} and the references therein) and because \kwl with small
dimension $k$ is known to be applicable to many other natural graph classes (e.g., planar graphs \cite{KieferPS19}).
The recent paper \cite{Kluge24} investigates the round complexity of \wl on circulant graphs,
exploiting the close connections of the subject with intricate mathematical concepts. 
Circulant graphs are also interesting on their own right as they naturally appear and are intensively
investigated in many other theoretical and applied areas; see, e.g., the books~\cite{BoseK19,Davis94,DobsonMM22}.
After all, our primary motivation for the study of circulant graphs is that this is the most natural first choice
of a graph class for benchmarking of combinatorial refinement in the realm of vertex-transitive graphs.

\subsection{Our contribution: Random circulant graphs}

We begin by fixing the basic notation and terminology.
Throughout the paper, the isomorphism relation $X\cong Y$ for graphs $X$ and $Y$ refers
to the standard combinatorial notion of graph isomorphism, regardless of any underlying
algebraic structure on the vertex sets of $X$ and $Y$. In particular, when speaking of
isomorphisms and automorphisms of graphs, we always mean the standard graph-theoretic concepts.

For a bijective vertex labeling $\lambda\function{V(X)}{\{0,1,\ldots,n-1\}}$ of an $n$-vertex digraph $X$,
let $X^\lambda$ denote the relabeled version of $X$, that is, the digraph
on vertex set $\{0,1,\ldots,n-1\}$ containing an edge from $x$ to $y$ whenever $X$ contains
an edge from $\lambda^{-1}(x)$ to $\lambda^{-1}(y)$. Given an input digraph $X$,
a \textbf{canonical labeling algorithm} produces a \emph{canonical labeling} $\lambda_X$ of $X$,
which satisfies the following properties:
\begin{itemize}
\item
  a labeling $\lambda_Y$ is computed by the algorithm also for every $Y\cong X$, and 
\item
  $X^{\lambda_X}=Y^{\lambda_Y}$ for all such~$Y$.
\end{itemize}
When applied to a randomly chosen $X$, such an algorithm may occasionally fail,
i.e., terminate without producing any output (this failure occurs simultaneously
for all isomorphic inputs). We say that the algorithms \emph{succeeds} with
probability at least $1-\varepsilon(n)$, if the failure probability on the inputs
with $n$ vertices does not exceed~$\varepsilon(n)$.

\paragraph{Canonical labeling of a random circulant.}

Our goal is to show that the individualization-refinement approach can be used to canonize
almost all circulants at minimal computational costs. Our treatment covers also
circulant directed graphs, which is advantageous for expository purposes as the case of digraphs
is technically somewhat simpler.
We show that the individualization
of a single vertex suffices for random circulant digraphs,
and that two individualized vertices are enough in the undirected case
(in fact, we just perform color refinement twice, each time with a single individualized vertex).
In both the directed and undirected cases, we maintain an overall running time of
$O(n^2\log n)$,
which is the standard running time of color refinement \cite{CardonC82}.
This is possible because our input graphs are vertex-transitive, and hence,
it does not actually matter which vertex is individualized --- even though, in the undirected case,
the choice of the second vertex to be individualized is not arbitrary.
Here, $n$ denotes the number of vertices. Thus, our time bound is actually linearithmic,
that is, it is $O(N\log N)$ for the input length $N$ where the input (di)graph is presented by the adjacency matrix
and the cyclic structure is not explicitly given (see below the discussion of different
representation concepts).
Note also that, as one might expect, our average-case bound of $O(n^2 \log n)$ is substantially better than
the worst-case bound resulting from \cite{EvdokimovP04}.\footnote{The algorithm in \cite{EvdokimovP04}
  involves the 2-dimensional Weisfeiler-Leman algorithm, which has time complexity $O(n^3 \log n)$.
  Its overall running time is stated as $n^c$, where the unspecified constant $c$ depends on the complexity
  of several algorithms from computational group theory.}
We summarize our main result in a somewhat condensed form as follows.

\begin{theorem}[Main Theorem]\label{thm:main}
A uniformly random circulant (di)graph with $n$ vertices is with probability at least $1-n^{-1/2+o(1)}$
canonizable by color refinement combined with vertex individualization in running time~$O(n^2\log n)$.  
\end{theorem}

Theorem \ref{thm:main} includes two statements, one for undirected graphs (where
all undirected circulant graphs are equiprobable) and the similar statement
for directed graphs. Note that the concept of a uniform distribution of $n$-vertex circulants
is somewhat ambiguous as the notion of an \emph{$n$-vertex circulant} alone
can be defined in three different natural ways:
\begin{itemize}
\item
  as a Cayley (di)graph of the cyclic group $\bZ_n$,
\item
  as an isomorphism class of Cayley (di)graphs of $\bZ_n$ (which we call an \emph{unlabeled circulant}),
\item
  as a (di)graph on the vertex set $\{0,1,\ldots,n-1\}$ isomorphic
  to a Cayley (di)graph of $\bZ_n$ (which we call a \emph{labeled circulant}).  
\end{itemize}
More formally, a \emph{connection set} $S \subseteq \bZ_n \setminus \{0\}$ defines the \emph{Cayley digraph} $\cay(\bZ_n, S)$,
where two vertices $x, y \in \bZ_n$ form a directed edge $(x, y)$ if $y - x \in S$. If $S$ is inverse-closed,
that is, $S = -S$, then $\cay(\bZ_n, S)$ is
undirected. When referring to a random Cayley (di)graph, we assume that all connection sets (with $S = -S$ in the undirected case)
are equally likely. Detailed definitions of the other two distributions are provided in Section~\ref{s:main}.

We first prove Theorem \ref{thm:main} for random Cayley digraphs and graphs (see the proof outline in Section \ref{s:proof-overview}), and then extend the result to the other two concepts. The transition from one distribution
to another is quite general and is based on known results in algebraic graph theory~\cite{BhoumikDM14,DobsonSV16,Muzychuk04}.

\paragraph{Canonical Cayley representation.}

Our canonical labeling algorithm in Theorem \ref{thm:main} is
based on individualization of a single vertex in the input digraph
followed by color refinement (or \WL1 in other terminology).
The aforementioned 2-dimensional Weisfeiler-Leman algorithm (\wl)
is strictly stronger than the combination of \WL1 with one-vertex individualization
(cf.~\cite[Theorem 3.2]{RattanS23}). It turns out
that \wl can be used to solve an even more challenging algorithmic problem
than computing a canonical labeling, which we describe below.

Theorem \ref{thm:main} provides an algorithm for producing a canonical labeling $\lambda_X$ of a
random circulant $X$. Note, however, that the proof does not guarantee
that the canonical form $X^{\lambda_X}$ is a Cayley digraph
of $\bZ_n$---or, equivalently, that the cycle $(\lambda_X^{-1}(0),\lambda_X^{-1}(1),\ldots,\lambda_X^{-1}(n-1))$
is an automorphism of $X$. In cases where this condition holds, i.e., when
$X^{\lambda_X}=\cay(\bZ_n,S)$ for some connection set $S$, we say that
the map $\lambda_X$ is a \emph{Cayley representation} of $X$.
We are interested in an algorithm which, with high probability, succeeds
in computing a map $\lambda_X$ that is \emph{both} a canonical labeling
and a Cayley representation of~$X$.

\begin{theorem}\label{thm:CCR}
  A uniformly random labeled circulant admits a canonical Cayley representation
  computable by means of \wl in time $O(n^3\log n)$ with success probability~$1-O(n^22^{-n/8})$.  
\end{theorem}

Compared to Theorem~\ref{thm:main}, the proof of Theorem~\ref{thm:CCR} requires less technical effort,
primarily due to the greater expressive power of the \wl algorithm.
Moreover, this algorithm is closely related to the notion of a coherent configuration in algebraic
combinatorics \cite{CP2019}.
This connection enables us to leverage strong results from algebraic graph theory~\cite{BhoumikDM14,DobsonSV16,EvdokimovP03},
whose algorithmic interpretation is the core of Theorem~\ref{thm:CCR}.

\section{Proof strategy and further implications}\label{s:proof-overview}

\subsection{Proof overview of the Main Theorem: The power of walk counts}

Remarkably, we show that canonization of a random circulant does not even need the
full power of color refinement and can actually be accomplished by a weaker algorithmic tool.
Let $G$ be an arbitrary (di)graph on the vertex set $V=\{0,1,\ldots,n-1\}$,
and let $T\subseteq V$. The \emph{walk matrix} $W_T=(w_{ik})_{i,k\in V}$ is defined
by setting $w_{ik}$ to be the number of walks of length $k$ from the vertex $i$ to
a vertex in $T$. That is, $w_{ik}$ is the number of vertex sequences $x_0,x_1,\ldots,x_k$
in $G$ such that $x_0=i$, $x_k\in T$, and the pair $(x_j,x_{j+1})$ is an edge of $G$ for every $j<k$.
If $A$ is the adjacency matrix of $G$ and $\chi_T$ denotes the characteristic
vector of the subset $T$, then $W_T$ is formed by the columns $\chi_T,A\chi_T,A^2\chi_T,\ldots,A^{n-1}\chi_T$.
The theory of walk matrices, including their applicability to isomorphism testing,
has been developed by Godsil \cite{Godsil12} and by Liu and Siemons \cite{LiuS22}.
Let $G_T$ be obtained from $G$ by coloring all vertices in $T$ by the same color.
We call $G_T$ \emph{walk-discrete} if the rows of $W_T$ are pairwise different.
For any walk-discrete $G_T$, the walk matrix $W_T$ yields a canonical labeling of the vertices of $G_T$.
This purely algebraic canonization method can be superseded by the purely combinatorial
method of color refinement because if $w_{uk}\ne w_{vk}$ for some $k$, then color refinement assigns
distinct colors to the vertices $u$ and $v$ in $G_T$ (see Section~\ref{ss:CR} for details).

Let $W=W_V$ be the \emph{standard} walk matrix of $G=G_V$. Obviously,
$G$ is walk-discrete whenever $W$ is non-singular.
Noteworthy, the rank of $W$ for an undirected graph $G$ is equal to the number of different
\emph{main} eigenvalues of the adjacency matrix~$A$; see \cite{Hagos02}.
As shown by O'Rourke and Touri \cite{ORourkeT16}, a random undirected graph $\rgraph$
has non-singular walk matrix with high probability.
As a consequence, $\rgraph$ is, with high probability, canonizable by
computing its standard walk matrix.\footnote{In fact, only a small part of the walk matrix
  suffices for this purpose --- as shown in \cite{ESA23}, $\rgraph$ is canonizable with high probability
  by assigning each vertex $u$ the triple $(w_{u1}, w_{u2}, w_{u3})$. We also refer the interested reader
  to \cite{Kriege22} for applications of color refinement and walk numbers in machine learning.}

The above theory essentially exploits the fact that the adjacency matrices of
undirected graphs are symmetric and, by this reason, does not apply to directed graphs.
Nevertheless, we obtain the following spectral criterion for circulant digraphs.

\begin{lemma}\label{lem:simple}
  Let $X$ be a Cayley digraph of a cyclic group and $X_0=X_{\{0\}}$
  be its version with one individualized vertex.
Let $W_0$ be the walk matrix of $X_0$.
  Then $W_0$ is non-singular (implying that $X_0$ is walk-discrete)
  if and only if $X$ has simple spectrum, that is, all eigenvalues of $X$
  are pairwise distinct.
\end{lemma}

Suppose now that $X$ is an undirected Cayley graph of $\bZ_n$.
In this case, the map $x\mapsto(-x)\bmod n$ is an automorphism of $X_0$,
which implies that the walk matrix of $X_0$ has at most $\lceil(n+1)/2\rceil$
different rows. If this bound is achieved, we call $X_0$ \emph{walk-saturated}.
On the other hand, the spectrum of $X$ has at most $\lceil(n+1)/2\rceil$
different eigenvalues. If there are exactly so many eigenvalues,
we say that $X$ has \emph{saturated spectrum}. We have the following analog of Lemma~\ref{lem:simple}
for the undirected case.

\begin{lemma}\label{lem:saturated}
  Let $X$ be a Cayley graph of the cyclic group $\bZ_n$. Then $W_0$ has the maximum possible rank $\lceil(n+1)/2\rceil$
  (implying that $X_0$ is walk-saturated)
  if and only if $X$ has saturated spectrum.
\end{lemma}

As an immediate consequence of Lemma \ref{lem:simple},
the rows of the walk matrix yield a canonical labeling of a circulant digraph whenever it has a simple spectrum.
With only a small amount of additional technical effort (see Lemma \ref{lem:w-saturated} in the next section),
Lemma \ref{lem:saturated} implies that the walk matrix can also be used
to canonize a circulant graph whenever it has a saturated spectrum. The following theorem therefore
estimates the success probability of these canonization methods on random circulants. 

\begin{theorem}\label{thm:whp}\hfill
  \begin{enumerate}[\bf 1.]
  \item
    If $S\subseteq\bZ_n\setminus\{0\}$ is chosen uniformly at random,
    then the Cayley digraph $\cay(\bZ_n,S)$ has simple spectrum with probability at least $1-n^{-1/2+o(1)}$.
  \item
    If $S\subseteq\bZ_n\setminus\{0\}$ is chosen uniformly at random among all inverse-closed sets,
    then the Cayley graph $\cay(\bZ_n,S)$ has saturated spectrum with probability at least $1-n^{-1/2+o(1)}$.    
  \end{enumerate}
\end{theorem}

Thus, Theorem \ref{thm:main} for Cayley (di)graphs follows from Theorem \ref{thm:whp} in view of Lemmas \ref{lem:simple} and \ref{lem:saturated}.
Theorem \ref{thm:whp}---and hence, as already noted, also Theorem \ref{thm:main}---extends to the
other two models of a random circulant discussed in Section \ref{s:intro}, namely, the random labeled and unlabeled circulants.

Theorem \ref{thm:whp} is our main technical contribution and may be of independent interest in the context of research
on random circulant matrices~\cite{BoseK19,Meckes09}.
Moreover, this result has further noteworthy consequences for properties of random circulant graphs,
which we discuss in the next subsection.

\subsection{Compactness and naïve canonization}

Note that two $n$-vertex graphs $G$ and $H$ with adjacency matrices $A$ and $B$, respectively,
are isomorphic if and only if there is an $n\times n$ permutation matrix $P$ such that $AP=PB$.
This observation leads to a natural linear programming relaxation for the graph isomorphism problem.
Recall that an $n\times n$ real matrix $S$ is \emph{doubly stochastic} if its
elements are nonnegative and all its rows and columns sum up to~$1$.
A doubly stochastic matrix $S$ satisfying the equation $AS=SB$ is called a \emph{fractional isomorphism}
from $G$ to $H$. In particular, if $AS=SA$, then $S$ is a \emph{fractional automorphism} of~$G$.

Tinhofer \cite{Tinhofer86} calls a graph $G$ \emph{compact} if the polytope formed in $\reals^{n^2}$ by the
fractional automorphisms of $G$ is integral, that is, the extreme points of this polytope have integer coordinates.
If a compact graph $G$ is isomorphic to another graph $H$, then the
polytope of fractional isomorphisms from $G$ to $H$ is also integral.
On the other hand, this polytope contains no integral point at all if $G$ and $H$ are non-isomorphic.
This has an algorithmic consequence. If we know that a graph $G$ is compact, then we can decide
whether or not $G$ is isomorphic to any other given graph $H$ in polynomial time by using linear
programming. It suffices to compute an extreme point of the polytope of fractional isomorphisms
from $G$ to $H$ and check if it has integer coordinates.

While no efficient method is known in general for determining whether a graph $G$ is compact,
Schreck and Tinhofer \cite{SchreckT88} established a sufficient condition for the compactness
of circulant graphs. In our terminology, they showed that circulant graphs with saturated spectrum are compact.
Consequently, Part 2 of our Theorem \ref{thm:whp} immediately implies the following result.

\begin{corollary}\label{cor:compact}
Almost all circulant graphs are compact. More precisely,
for all three notions of a circulant graph, if $X$ is a uniformly random
circulant graph on $n$ vertices, then $X$ is compact with probability at least $1-n^{-1/2+o(1)}$.
\end{corollary}

In \cite{Tinhofer91}, Tinhofer presented another, fairly surprising combinatorial approach to testing
whether two graphs are isomorphic provided one of them is compact.
The approach can be recast as a canonization algorithm and
was considered as such also by Immerman and Lander \cite{ImmermanL90}.
Specifically, they consider the following algorithmic procedure.

\medskip

\noindent\textsc{Naïve canonization}

\smallskip

\noindent\textsc{Input:} a graph $X$.
\begin{enumerate}
\item
  Set $\widetilde X=X$.
\item
  Run color refinement on $\widetilde X$ and denote the resulting colored graph by $\widehat X$.
\item
  If the vertex colors in $\widehat X$ are pairwise distinct, then output this coloring as a canonical labeling of $X$.
\item
  Otherwise
  \begin{enumerate}
  \item
    choose an arbitrary vertex $v$ in the non-singleton color class of $\widehat X$ with least\footnote{%
      The set of colors produced by color refinement is endowed with a natural order.
      Moreover, we can suppose that each vertex of $\widehat X$ is colored by an integer in $\{0,1,\ldots,n-1\}$;
      see Section~\ref{ss:CR}.}
    color;
\item
reset $\widetilde X$ to be $\widehat X$ with the vertex $v$ individualized;
\item
  repeat Step 2 again.
\end{enumerate}
\end{enumerate}

We say that the above procedure \emph{works correctly} on input $X$ if it produces a canonical labeling
of $X$ irrespectively of which vertex $v$ is chosen in any execution of Step 4(a).
Note that naïve canonization works correctly whenever it terminates after the first execution of Step 3,
which occurs for almost all graphs $X$ by \cite{BabaiES80}. In the general case, when Step 4 is executed
at least once, it is clear that a non-backtracking refinement-individualization procedure like this
cannot be expected to be correct. Indeed, naïve canonization obviously fails on
any regular but non-vertex-transitive graph $X$. All the more surprising, then, is the result established in \cite{Tinhofer91}
that the approach still works correctly on such a broad and naturally defined class of graphs as the class of compact graphs.

Combining this fact with Corollary \ref{cor:compact}, we obtain the following result.

\begin{corollary}\label{cor:tinhofer}
Naïve canonization works correctly for almost all circulant graphs. More precisely,
this holds for all but a fraction of $n^{-1/2+o(1)}$ of circulant graphs on $n$ vertices,
under any of the three models of circulant graphs.
\end{corollary}

Note that the class of graphs on which naïve canonization succeeds is strictly larger than the class
of compact graphs; see \cite{ArvindKRV17}. This remains true also when we focus on circulant graphs.
Indeed, Schreck and Tinhofer \cite{SchreckT88} proved that a non-empty and non-complete circulant graph with a prime number of vertices
is compact if and only if it has saturated spectrum. Remarkably, Kluge~\cite{Kluge24} showed that
naïve canonization works correctly on \emph{every} circulant graph with a prime number of vertices.

Finally, we remark that naïve canonization makes sense only if we are a priori confident in its correctness.
Corollary \ref{cor:tinhofer} is quite constructive in this regard,
as the sufficient condition of Schreck and Tinhofer \cite{SchreckT88}---having a saturated spectrum---is
verifiable in polynomial time. Moreover, it can be shown by a direct argument that naïve canonization
works correctly on every walk-saturated circulant graph.

\subsection{Organization of the rest of the paper}
The proof of our Main Theorem (Theorem \ref{thm:main}) spans Sections \ref{s:wm-cr}--\ref{s:main}.
Section \ref{s:wm-cr} begins with preliminaries and then provides
a detailed description of our canonical labeling algorithm for circulant graphs with saturated spectrum.
Lemmas \ref{lem:simple} and \ref{lem:saturated} are proved in Section \ref{s:spectra},
as special cases of a more general result stated there as Lemma~\ref{lem:spectrum}.
Theorem \ref{thm:whp} is proved in Section \ref{s:whp}. This establishes Theorem \ref{thm:main}
for Cayley (di)graphs over~$\bZ_n$.

The proof of Theorem \ref{thm:main} is completed in Section~\ref{s:main},
which is devoted to the other two models of a random circulant. In this section,
we establish two ``transition'' lemmas. Specifically, Lemma \ref{lem:u} allows us to extend
the probability bound of Theorem \ref{thm:whp} from random Cayley (di)graphs to random
\emph{unlabeled} circulants, while Lemma \ref{lem:l} further extends this bound to random \emph{labeled} circulants.

Finally, Theorem \ref{thm:CCR} is proved in Section~\ref{s:wl2}.

\section{The walk matrix and color refinement}\label{s:wm-cr}

\subsection{Definitions and a relationship}\label{ss:CR}

Speaking of a directed graph or, for short, \emph{digraph} $G=(V,E)$,
we always assume that $G$ is loopless, that is, the adjacency relation $E\subset V^2$
is irreflexive. Without loss of generality
we suppose that $G$ is defined on the vertex set $V=\{0,1,\ldots,n-1\}$.
If $E$ is symmetric, $G$ is referred to as an (undirected) \emph{graph}.
The definitions given below for digraphs apply, as a special case, also to graphs.

For $t\in V$, we write $G_t$ to denote the digraph $G$ with distinguished vertex $t$.
The vertex $t$ is called \emph{terminal} or \emph{individualized}.
We consider $G_t$ to be a vertex-colored digraph where all vertices are equally colored
with the exception of $t$ which has a special color. An isomorphism from $G_t$ to
another vertex-individualized digraph $H_u$
is defined as a digraph isomorphism from $G$ to $H$ taking $t$ to~$u$.

A \emph{walk} in $G$ is a sequence of vertices $x_0x_1\ldots x_k$ such that
$(x_i,x_{i+1})\in E$ for every $0\le i<k$. We say that $x_0x_1\ldots x_k$
is a walk of length $k$ from $x_0$ to $x_k$. Note that a one-element sequence $x_0$
is a walk of length $0$. Given a digraph $G_t$ with terminal vertex $t$,
we define its \emph{walk matrix} $W_t=(w_{x,k})_{0\le x,k<n}$ by setting $w_{x,k}$
to be the number of walks of length $k$ from $x$ to $t$. Let
$$
W_t(x)=(w_{x,0},w_{x,1},\ldots w_{x,n-1})
$$
be the row of $W_t$ corresponding to the vertex $x$. If $\phi$ is an isomorphism from $G_t$ to $H_u$,
then clearly $W_u(\phi(x))=W_t(x)$. This means that $W_t(x)$ can be used as a canonical label
for a vertex $x$ in $G_t$. We call $G_t$ \emph{walk-discrete} if $W_t(x)\ne W_t(x')$ for all $x\ne x'$.
Thus, the walk matrix yields a canonical labeling for the class of walk-discrete digraphs
with an individualized vertex.

As it was mentioned in Section \ref{s:proof-overview}, the walk matrix is efficiently computable by
linear algebraic operations. For walk-discrete digraphs, the corresponding canonical labeling can also be
obtained combinatorially via the \emph{color refinement} algorithm (CR), which we now describe formally.

Given a vertex-colored digraph $G$ with initial coloring $C_0$, CR iteratively computes new colorings
\begin{equation}
  \label{eq:ref-step}
C_{i+1}(x)=\of{C_i(x),\Mset{C_i(y)}_{y\in N(x)}},  
\end{equation}
where $\Mset{}$ denotes a multiset and $N(x)=\Set{y}{(x,y)\in E}$ is the out-neighborhood of $x$.
For each $i$, the coloring $C_i$ is canonical, i.e., isomorphism-invariant.
Indeed, an easy induction on $i$ shows that if $\phi$ is a (color-preserving) isomorphism from $G$ to
a vertex-colored graph $H$, then
\begin{equation}
  \label{eq:col-canon}
C_i(\phi(x))=C_i(x).
\end{equation}

The color classes of $C_{i+1}$ refine the color classes of $C_i$:
if $C_i(x)\ne C_i(x')$, then $C_{i+1}(x)\ne C_{i+1}(x')$.
The algorithm terminates after performing $n$ refinement rounds, where $n$ is the number
of vertices in $G$, and outputs the coloring $C_n$. Note that the color partition becomes
stable by this point.

A relationship between CR and the walk counts was observed in \cite{PowersS82}.
We use the following adaptation of this result for our purposes.

\begin{lemma}\label{lem:CR}
  Let $t\in V(G)$ and $C_0$ be the vertex coloring associated with $G_t$, that is, $C_0(x)=C_0(x')$
and $C_0(x)\ne C_0(t)$ for all $x\ne t$ and $x'\ne t$.
  If $w_{x,k}\ne w_{x',k}$, then $C_k(x)\ne C_k(x')$.
\end{lemma}

\begin{proof}
  Using the induction on $k$, we prove that $w_{x,k}=w_{x',k}$ whenever $C_k(x)=C_k(x')$.
  In the base case of $k=0$, the equalities $w_{x,0}=w_{x',0}$ and $C_0(x)=C_0(x')$
  are equivalent by definition. Assume that $C_k(y)=C_k(y')$ implies $w_{y,k}=w_{y',k}$
  for all $y$ and $y'$. Let $C_{k+1}(x)=C_{k+1}(x')$. By the definition of the refinement step,
  we have $\Mset{C_k(y)}_{y\in N(x)}=\Mset{C_k(y)}_{y\in N(x')}$.
  Using the induction assumption, from here we derive the equality
  $\Mset{w_{y,k}}_{y\in N(x)}\allowbreak=\Mset{w_{y,k}}_{y\in N(x')}$.
  The equality $w_{x,k+1}=w_{x',k+1}$ now follows by noting that
  $w_{x,k+1}=\sum_{y\in N(x)}w_{y,k}$.
\end{proof}

Lemma \ref{lem:CR} shows that a canonical labeling of a walk-discrete digraph $G_t$
can be obtained by running CR on $G_t$ rather than by directly computing the walk matrix,
and we conclude this subsection by commenting on the efficiency of~CR.

The inductive definition \refeq{ref-step} leads to an exponential increase in the size of the color encoding.
To prevent this, colors are renamed after each refinement step.
In this way, we never need more than $n$ distinct color names.
To encode the colors after each refinement round, we can use, for example, the first $n$ non-negative integers in binary.
This allows us, in the next subsection, to refer to the \emph{least} color having a certain property.
Note that, once the renaming rule is fixed, the modified coloring remains canonical in the sense of
Equality~\refeq{col-canon}.

Finally, recall that CR can be implemented in time $O(n^2\log n)$ \cite{BerkholzBG17,CardonC82,ImmermanL90},
while preserving the canonicity of the final coloring with respect to graph isomorphism; see~\cite{BerkholzBG17}.

\subsection{Cayley graphs of a cyclic group}\label{ss:canon}

Let $\bZ_n$ denote a cyclic group of order $n$. More specifically, we let $\bZ_n=\Set{0,1,\ldots,n-1}$
and consider the addition modulo $n$ on this set.
The \emph{Cayley digraph} $X=\cay(\bZ_n,S)$ is defined by a \emph{connection set}
$S\subseteq\bZ_n\setminus\{0\}$ as follows: $V(X)=\bZ_n$ and $(x,y)\in E(X)$ if and only if
$y-x\in S$. Note that $S=N(0)$, the out-neighborhood of $0$.
If $S$ is \emph{inverse-closed}, i.e., $S=-S$, then $E(X)$ is symmetric and we speak of a \emph{Cayley graph}.
Cayley (di)graphs of $\bZ_n$ are also called \emph{circulant (di)graphs} or \emph{circulants}.

For $u\in\bZ_n$, let $X_u$ be the vertex-individualized version of $X$.
Since $X$ is vertex-transitive, all $X_u$ are isomorphic to each other,
and we can speak about $X_0$ without loss of generality.
Clearly, in order to canonize $X$, it is sufficient to canonize $X_0$.
Therefore, the canonization method in the preceding subsection applies to any
Cayley digraph $X=\cay(\bZ_n,S)$ provided that $X_0$ is walk-discrete.
We just have to individualize an arbitrary vertex of $X$ and then run~CR.

This method does not work for circulant \emph{graphs}.
Indeed, define $\rho\function{\bZ_n}{\bZ_n}$ by $\rho(x)=-x$.
If $S=-S$, then $\rho$ is an automorphism of $X=\cay(\bZ_n,S)$
and, hence, $W_0(\rho(x))=W_0(x)$. This implies that the walk matrix $W_0$
has at most $\lceil(n+1)/2\rceil$ different rows, and $X_0$ cannot be walk-discrete.
If this maximum is attained, we call $X_0$ \emph{walk-saturated}.

\begin{lemma}\label{lem:w-saturated}
Let $X=\cay(\bZ_n,S)$, where $S=-S$, and suppose that $X_0$ is walk-saturated.
Fix $u\ne0$ such that $u\ne n/2$ if $n$ is even. Then
$$
(W_0(x),W_u(x))\ne(W_0(y),W_u(y))
$$
for any two different vertices $x$ and $y$ of~$X$.
\end{lemma}

\begin{proof}
Since $X_0$ is walk-saturated, the equality $W_0(x)=W_0(y)$ for $x\ne y$
is possible only if $y=\rho(x)$, i.e., $y=-x$ in $\bZ_n$. Note that
$W_u(x)=W_0(x-u)$. Therefore, the equality $W_u(x)=W_u(y)$ implies
that $y=2u-x$. The equalities $y=-x$ and $y=2u-x$ can be fulfilled simultaneously
only if $2u=0$, which is excluded.
\end{proof}

Lemma \ref{lem:w-saturated} justifies the correctness of the following algorithm
for the class of walk-saturated circulant graphs.

\medskip

\noindent\textsc{Canonical labeling algorithm}

\smallskip

\noindent\textsc{Input:} a circulant graph $X$.
\begin{enumerate}
\item
  Individualize an arbitrary vertex of $X$. By vertex-transitivity, we can
  without loss of generality assume that the individualized vertex is~$0$.
\item
  Run CR on $X_0$. Let $C$ be the obtained coloring of the vertex set.
\item
  Let $c$ be the least color such that there are exactly two vertices
  $u_1$ and $u_2$ with $C(u_1)=C(u_2)=c$. If such $c$ does not exist, then give up.
  Let $u$ be any of $u_1$ and $u_2$. Individualize $u$ in~$X$.
\item
  Run CR on $X_u$. Let $C'$ be the obtained coloring.
\item
  To each vertex $x$, assign the label $L(x)=(C(x),C'(x))$.
\item
Check that all labels $L(x)$ are pairwise distinct. If not, then give up.
\end{enumerate}

For each circulant input graph, our canonization algorithm either produces
a vertex labeling $L$ or explicitly gives up (doing always the same
for isomorphic inputs). The labeling $L$ is canonical because it does not
depend on which vertex is individualized in Step 1 (by vertex-transitivity)
nor in Step 3 (because $u_1$ and $u_2$ are interchangeable by an automorphism
of $X_0$). Lemma \ref{lem:w-saturated} ensures that the algorithm succeeds whenever $X_0$ is walk-saturated,
and this will allow us to estimate the success probability based on
Lemma \ref{lem:saturated} and Theorem~\ref{thm:whp}.

Finally, we remark that if the algorithm is run on a non-circulant input graph
and outputs a vertex labeling, then this labeling does not need to be canonical.
To make it canonical in all cases, Steps 1 and 3 have to be performed for
all possible individualized vertices, which can yield $2n$ different labelings
$L_1,\ldots,L_{2n}$. Of all these candidate labelings, the algorithm chooses
that which yields the isomorphic copy of $X$ with lexicographically least adjacency
matrix. The running time of this algorithm variant is $O(n^3\log n)$.
The similar modification is as well possible in the case of digraphs.

\section{The walk matrix and the spectrum of circulants}\label{s:spectra}

\subsection{The spectrum of a circulant}\label{ss:spectrum}

Let $A=(a_{ij})$ be the adjacency matrix of a circulant digraph $X$,
that is, $A$ is the 0--1 matrix whose rows and columns are indexed
by the elements $0,1,\ldots,n-1$ of $\bZ_n$ such that $a_{ij}=1$ exactly when
$(i,j)\in E(X)$, that is, $j-i\in S$. Note that $A$ is a \emph{circulant matrix},
which means that the $(i+1)$-th row of $A$ is obtained from its $i$-th row by the cyclic shift
in one element to the right.

Let $\omega$ be an $n$-th root of unity, i.e., $\omega\in\bC$ and $\omega^n=1$.
For the vector $V_\omega=(1,\omega,\omega^2,\ldots,\omega^{n-1})^\top$,
the definition of a circulant matrix easily implies the equality
$$
A\,V_\omega=\of{a_0+a_1\omega+a_2\omega^2+\cdots+a_{n-1}\omega^{n-1}}\,V_\omega=\of{\sum_{j\in S}\omega^j}V_\omega,
$$
where $(a_0,a_1,a_2,\ldots,a_{n-1})$ is the first row of $A$, that is,
the characteristic vector of $S\subset\bZ_n$.
We conclude that $V_\omega$ is an eigenvector of $A$ corresponding to the
eigenvalue $\lambda_{\omega,S}=\sum_{j\in S}\omega^j$.

Now, let $\omega=\zeta_n$ be a primitive $n$-th root of unity. To be specific, we fix $\zeta_n=e^{-2\pi\mathrm{i}/n}$.
The $n$ column vectors $V_\omega$ for $\omega=\zeta_n^0,\zeta_n^1,\zeta_n^2,\ldots,\zeta_n^{n-1}$ form
a Vandermonde matrix with non-zero determinant. It follows that these $n$ vectors are linearly independent
and, therefore, $\lambda_{\zeta_n^0,S},\lambda_{\zeta_n^1,S},\ldots,\lambda_{\zeta_n^{n-1},S}$
is the full spectrum of $A$. The $i$-th eigenvalue in this sequence will be below denoted by
\begin{equation}\label{eq:liS}
\lambda_{i,S}=\sum_{j\in S}\zeta_n^{ij}.  
\end{equation}

\subsection{Discrete Fourier transform}

Let $\bC^{\bZ_n}$ denote the vector space of all functions from ${\bZ_n}$
to the field of complex numbers $\bC$ with pointwise addition and pointwise scalar multiplication.
The pointwise multiplication on $\bC^{\bZ_n}$ will be denoted by $\circ$.
Another way to introduce a product on $\bC^{\bZ_n}$ is to consider the convolution
$\alpha*\beta$ of two functions $\alpha,\beta\function {\bZ_n}{\bC}$, which is defined by
$(\alpha*\beta)(x)=\sum_{y\in {\bZ_n}}\alpha(x-y)\beta(y)$ for each $x\in {\bZ_n}$.
Both $\circ$ and $*$ are bilinear and, therefore, both $(\bC^{\bZ_n},\circ)$ and $(\bC^{\bZ_n},*)$
are $n$-dimensional algebras over $\bC$.
The algebra $(\bC^{\bZ_n},*)$ can alternatively be seen as the \emph{group algebra} of
$\bZ_n$ over $\bC$ and, as such, it is semisimple by Maschke's theorem; see, e.g., \cite[Section 7.1]{DrozdK94}.
Like any two $n$-dimensional commutative semisimple $\bC$-algebras, the algebras
$(\bC^{\bZ_n},*)$ and $(\bC^{\bZ_n},\circ)$ are isomorphic (see \cite[Corollary 2.4.2]{DrozdK94}).
We now describe an explicit algebra isomorphism from $(\bC^{\bZ_n},*)$ to $(\bC^{\bZ_n},\circ)$.

For $T\subseteq {\bZ_n}$, let $\chi_T\in\bC^{\bZ_n}$ be the characteristic function of $T$.
In particular, $\chi_{\bZ_n}$ is the identically one function.
For $x\in {\bZ_n}$, we set $\delta_x=\chi_{\{x\}}$.

The \emph{discrete Fourier transform (DFT)} is the linear operator $\fourier\function{\bC^{\bZ_n}}{\bC^{\bZ_n}}$
mapping a function $\alpha\function{\bZ_n}{\bC}$ into the function $\fourier(\alpha)\function{\bZ_n}{\bC}$ defined by
\begin{equation}\label{eq:dft}
\fourier(\alpha)(i)=\sum_{j=0}^{n-1}\zeta_n^{ij}\alpha(j).  
\end{equation}
In the standard basis $\delta_0,\delta_1,\ldots,\delta_{n-1}$,
the DFT is represented by the matrix $F=(\zeta_n^{ij})_{i,j\in\bZ_n}$.
Since $F$ is the familiar Vandermonde matrix with non-zero determinant,
the map $F$ is a linear isomorphism from $\bC^{\bZ_n}$
onto itself. It is well known and easy to derive from the definitions that
\begin{equation}\label{eq:op-op}
 \fourier(\alpha*\beta)=\fourier(\alpha)\circ\fourier(\beta). 
\end{equation}

\subsection{The rank of the walk matrix}\label{ss:rank}

We are now prepared to derive Lemmas \ref{lem:simple} and \ref{lem:saturated}
from the following more general fact.

\begin{lemma}\label{lem:spectrum}
$X=\cay(\bZ_n,S)$ has exactly $\rank W_0$ distinct eigenvalues, where $W_0$ is the walk matrix of~$X_0$.  
\end{lemma}

\begin{proof}
Let $X'=\cay(\bZ_n,-S)$ be the transpose of the digraph $X$, i.e., the digraph obtained
from $X$ by reversing all arcs. We begin by noting that $X$ and $X'$ have the same number
of distinct eigenvalues. To see this, consider the automorphism $f$
of the cyclotomic field $\bQ(\zeta_n)$---the field obtained by adjoining $\zeta_n$ to the field of
rationals $\bQ$---which maps $\zeta_n$ to $\zeta_n^{-1}$ while fixing $\bQ$.
It suffices to observe that $f(\lambda_{i,S})=\lambda_{i,-S}$.
Denote the number of distinct eigenvalues of $X'$, and hence also of $X$, by~$R$.
  
A column vector $(a_0,a_1,\ldots,a_{n-1})^\top$ will be naturally identified with
the function $\alpha\in\bC^{\bZ_n}$ defined by $\alpha(x)=a_x$ for all $x\in\bZ_n$.
In this way, the columns of the walk matrix $W_0$ correspond to the functions
$\eta_0,\eta_1,\ldots,\eta_{n-1}$ where $\eta_k(x)=w_{x,k}$. 
Thus, the rank of $W_0$ is equal to the dimension of the linear subspace $U$
of $\bC^{\bZ_n}$ spanned by these functions. We, therefore, have to prove that $\dim U=R$.

Note that
\begin{multline*}
  \eta_{k+1}(x)=\sum_{y\in N(x)}\eta_k(y)=\sum_{y\in\bZ_n}\chi_S(y-x)\eta_k(y)\\
  =\sum_{y\in\bZ_n}\chi_{-S}(x-y)\eta_k(y)=(\chi_{-S}*\eta_k)(x).  
\end{multline*}
It follows that $\eta_0=\delta_0$, $\eta_1=\chi_{-S}$, $\eta_2=\chi_{-S}*\chi_{-S}$
and, generally, $\eta_k=\chi_{-S}^{*(k)}$ is the $(k-1)$-fold convolution of $k$
copies of the characteristic function $\chi_{-S}$ of the set~$-S$.

Let us apply the discrete Fourier transform $\fourier$.
Note that $\fourier(\delta_0)$ is the all-ones vector.
As easily seen from \refeq{liS} and \refeq{dft}, $\fourier(\chi_{-S})$ is
the vector whose entries are the eigenvalues $\lambda_{0,-S},\lambda_{1,-S},\ldots,\lambda_{n-1,-S}$
of the transpose $X'=\cay(\bZ_n,-S)$. Equality \refeq{op-op} readily
implies that the matrix formed by the column vectors $\fourier(\eta_0),\fourier(\eta_1),\ldots,\fourier(\eta_{n-1})$
has exactly $R$ distinct rows. Consequently, $\dim U=\dim\fourier(U)\le R$.
In fact, equality holds because the aforementioned matrix contains
a non-generate Vandermonde matrix of size $R\times R$ as a submatrix.
\end{proof}

Lemmas \ref{lem:simple} and \ref{lem:saturated} correspond, respectively, to the special cases
$\rank W_0=n$ and $\rank W_0=\lceil(n+1)/2\rceil$ of Lemma \ref{lem:spectrum}.
If $S=-S$, then $\rank W_0\le\lceil(n+1)/2\rceil$ due to the symmetry of the undirected graph $X=\cay(\bZ_n,S)$.
Lemma \ref{lem:spectrum} shows that $X$ can have at most $\lceil(n+1)/2\rceil$ distinct eigenvales\footnote{%
This can be seen also directly as Equality \refeq{liS} implies that
$\lambda_{a,S}=\lambda_{b,S}$ for $a\ne b$ whenever $a+b=n$.},
and that this maximum is attained exactly when 
$\rank W_0$ attains the same maximum value $\lceil(n+1)/2\rceil$.

\section{Proof of Theorem \ref{thm:whp}}\label{s:whp}

We set $\zeta_n=e^{-2\pi\mathrm{i}/n}$.
As discussed in Subsection \ref{ss:spectrum}, a circulant $X=\cay(\bZ_n,S)$
has eigenvalues $\lambda_0,\lambda_1,\ldots,\lambda_{n-1}$ where
$$
\lambda_a=\sum_{j\in S}\zeta_n^{aj}=\sum_{j=0}^{n-1}\chi_S(j)\zeta_n^{aj}.
$$
Let $\sigma_j=\chi_S(j)$. If $X$ is a random digraph, i.e., the connection set $S$ is chosen uniformly at random
among all subsets of $\bZ_n\setminus\{0\}$, then $\sigma_1,\sigma_2,\ldots,\sigma_{n-1}$ is a Bernoulli process,
that is, these $n-1$ random variables are independent and identically distributed
with $\sigma_j$ taking each of the values 0 and 1 with probability $1/2$.
If $X$ is a random graph, i.e., the connection set $S=-S$ is chosen randomly among all inverse-closed subsets,
then the values $\sigma_1,\sigma_2,\ldots,\sigma_{\lfloor n/2\rfloor}$
form a Bernoulli process, and the remaining values are determined by the equality $\sigma_j=\sigma_{n-j}$.
For each $a=0,1,\ldots,n-1$, the eigenvalue
\begin{equation}
  \label{eq:lambdaa}
\lambda_a=\sum_{j=1}^{n-1}\sigma_j\zeta_n^{aj}  
\end{equation}
becomes a random variable taking its values in the cyclotomic field~$\bQ(\zeta_n)$.

We will use the following observation.
As usually, $\phi(n)$ stands for Euler's totient function.

\begin{lemma}\label{lem:cyclotomic}
No two different subsets of $\Set{\zeta_n^j}{1\leq j\leq n/\ln n}$ have equal sums of their elements.
\end{lemma}

\begin{proof}
The known lower bounds for $\phi(n)$ (see, e.g., \cite[Thm.~8.8.7]{BachS96}) imply that
$\phi(n) > n/\ln n$ for $n\ge3$. The existence of two different subsets with equal sums would therefore yield
a non-trivial linear combinations with rational coefficient of 
$1,\zeta_n,\zeta_n^2,\ldots,\zeta_n^{\phi(n)-1}$, contradicting the fact that these numbers
form a basis of $\bQ(\zeta_n)$ considered as a vector space over~$\bQ$.
\end{proof}

Our overall strategy for proving Theorem \ref{thm:whp} will be to use the union bound
\begin{equation}
  \label{eq:union-bound}
\prob{\lambda_a=\lambda_b\text{ for some }0\le a,b\le n-1}\le
\sum_{0\le a,b\le n-1}\prob{\lambda_a=\lambda_b}    
\end{equation}
and to show that the right hand side is bounded by $n^{-1/2+o(1)}$.
The summation ranges over unequal $a$ and $b$;
in the undirected case, we additionally require $a+b\ne n$.
For a fixed pair $\{a,b\}$, we have to show that $\lambda_a=\lambda_b$ occurs
with sufficiently small probability. Using \refeq{lambdaa}, this equality can be rewritten as
\begin{equation}\label{eq:k=m}
\sum_{j=1}^{n-1}\sigma_j\zeta_n^{aj}=\sum_{j=1}^{n-1}\sigma_j\zeta_n^{bj}
\end{equation}
The basic idea is to ``distill'' a sufficiently large set of indices $J\subset\{1,\ldots,n-1\}$ such that
exposing all random variables $\sigma_j$ for $j\notin J$ converts \refeq{k=m} into an equality
\begin{equation}
  \label{eq:exposed}
\sum_{j\in J}\sigma_j\zeta_n^{aj}=\sum_{j\in J}\sigma_j\zeta_n^{bj}+\mathrm{const}  
\end{equation}
that can be satisfied by at most one assignment to the remaining random variables $\sigma_j$ for $j\in J$.
The last condition immediately implies the bound $\prob{\lambda_a=\lambda_b}=2^{-|J|}$,
which will be strong enough for our purposes. To justify that \refeq{exposed} has at most one satisfying assignment,
we will crucially rely on Lemma~\ref{lem:cyclotomic}.

Let
$$
\bU=\Set{z\in\bC}{|z|=1}
$$
denote the unit circle in the complex plane.
Lemma~\ref{lem:cyclotomic} will be applicable not only when $aj$ and $bj$ (modulo $n$) do not
exceed $n/\ln n$ for all $j\in J$---which would be technically difficult to ensure---but also whenever
the set $\Set{\zeta_n^{aj}}_{j\in J}\cup\Set{\zeta_n^{bj}}_{j\in J}$ is contained
within an arbitrary arc of $\bU$ of length $2\pi/\ln n$. In such cases, we can ``rotate''
this set by multiplying both sides of \refeq{exposed} by a suitable root of unity $\zeta_n^t$.
This will transform \refeq{exposed} into a linear combination of the complex numbers $\zeta_n^{aj+t}$ and
$\zeta_n^{bj+t}$ with exponents (modulo $n$) lying in the range covered by Lemma~\ref{lem:cyclotomic}.

The choice of a suitable set $J$ depends on specific properties of the pair $\{a,b\}$.
We divide all such pairs into three categories and describe an appropriate ``distillation''
of $J$ separately for each of the three cases. The relevant properties of a pair $\{a,b\}$
are expressed in terms of elementary number-theoretic parameters, which we now introduce.

Given $z\in\bC$, we write $\langle z\rangle$ to denote the cyclic subgroup of
the multiplicative group $\bC^\times$ generated by $z$. For an integer $c$, let $g(c)=|\langle\zeta_n^c\rangle|$
denote the order of the element $\zeta_n^c$ in the group $\langle\zeta_n\rangle$.
Note that $g(c)=n/\gcd(c,n)$. We also define
$$
c'=c/\gcd(c,n),
$$
noting that
\begin{equation}
  \label{eq:znzg}
  \zeta_n^{c}=\zeta_{g(c)}^{c'}.  
\end{equation}
The above notation will be used both for $c=a$ and $c=b$. We set
$$
g=g(a)\text{ and }h=g(b),
$$
and suppose, without loss of generality, that
$$
h\le g.
$$

Since $\gcd(a',g)=1$, the integer $a'$ can be regarded as an invertible element
of the ring $\bZ_g$. Let $r$ denote the inverse of $a'$ in $\bZ_g$, i.e., $r$ is the smallest
integer for which $ra'=1\pmod g$.

Note that $\langle\zeta_n^a\rangle=\langle\zeta_g\rangle$. Indeed, $\zeta_n^a\in\langle\zeta_g\rangle$
by \refeq{znzg}. On the other hand, \refeq{znzg} also implies that
$\zeta_g=\zeta_g^{ra'}=\zeta_n^{ar}$ belongs to $\langle\zeta_n^a\rangle$.
Set
\begin{equation}\label{eq:xi-eta-def}
\xi=\zeta_g=\zeta_n^{ar}\text{ and }\eta=\zeta_n^{br}.
\end{equation}
Thus, all $g-1$ non-unity elements of the multiplicative group $\langle\zeta_n^a\rangle=\langle\xi\rangle$
can be listed as
\begin{equation}
  \label{eq:xi}
\xi=\zeta_n^{ar},\ \xi^2=\zeta_n^{2ar},\ \ldots\ ,\xi^{g-1}=\zeta_n^{(g-1)ar}.
\end{equation}
It is useful to note that they appear in the left hand side of Equality \refeq{k=m}
for the indices $j=r,2r,\ldots,(g-1)r$, which are understood modulo~$n$.
The same positions in the right hand side of Equality \refeq{k=m} are occupied by
\begin{equation}
  \label{eq:eta}
\eta=\zeta_n^{br},\ \eta^2=\zeta_n^{2br},\ \ldots\ ,\eta^{g-1}=\zeta_n^{(g-1)br}.  
\end{equation}

The importance of the parameter $g$ stems from the fact that the elements of \refeq{xi},
together with the unity $\xi^{g}=1$, partition the unit circle $\bU$ into $g$ arcs, each of length $2\pi/g$.
When $g$ is large, this gives us better chances to obtain a linear combination \refeq{exposed}
such that the degrees $\zeta_n^{aj}$ involved in \refeq{exposed} are sufficiently close to one another on $\bU$,
allowing us to apply Lemma~\ref{lem:cyclotomic} as described above.

The minimum distance in $\bU$ between two distinct elements $\zeta_n^{brk_1}=\eta^{k_1}$
and $\zeta_n^{brk_2}=\eta^{k_2}$ is important by the same reason. This distance is equal to $2\pi/g'$ where
$$
g'=|\langle\eta\rangle|
$$
is the order of the element $\eta$ in $\langle\zeta_n\rangle$. Since $\langle\eta\rangle\subseteq\langle\zeta_n^b\rangle$,
we have
$$
g'\le h\le g.
$$
While all numbers in \refeq{xi} are distinct, the number of distinct numbers in the sequence \refeq{eta}
is equal to $g'$ if $g'<g$ and to $g'-1$ if $g'=g$.

The parameters $g$ and $g'$ are crucial for our analysis, and we partition all pairs $\{a,b\}$
into three groups based on the values of $g$ and $g'$. Recall that we assume
$g(b)\le g(a)$. In cases where $g(b)=g(a)$, we will further assume that $b<a$.
This additional assumption does not play any substantive role in the analysis
but allows us to treat $a,b$ as an ordered pair.

We now outline our argument for each of the three cases---Cases A, B, and C---described below.
After presenting the main ideas, we will proceed to a detailed proof, organized into
Claims A, B, and C, respectively.

\Case A{$g$ and $g'$ are large.}%
This is the most technically demanding case, and it is further divided into three subcases.
In the first two of them, we can find a sufficiently large set $K\subseteq\{1,\ldots,n-1\}$,
specifically of size $|K|=\Omega(\ln^2n)$, such that:
\begin{enumerate}[(i)]
\item
  the subsequence $(\xi^k)_{k\in K}$ of the sequence \refeq{xi} consists of distinct elements,
  and the same holds true for the corresponding subsequence $(\eta^k)_{k\in K}$ of~\refeq{eta};
\item
  the sets $\Set{\xi^k}_{k\in K}$ and $\Set{\eta^k}_{k\in K}$ are disjoint;
\item
  both sets are contained within a sufficiently short arc of $\bU$, specifically of length less than $1/\ln n$.
\end{enumerate}
We now argue that properties (i)--(iii) allow us to bound the probability of \refeq{k=m}
based on Lemma~\ref{lem:cyclotomic}.

Let us expose the values of all random variables $\sigma_j$ excepting $\sigma_{kr}$ with $k\in K$
(recall that the indices are considered modulo $n$). Equality \refeq{k=m} can now be written as
\begin{equation}
  \label{eq:c1c2}
c_1+\sum_{k\in K}\sigma_{kr}\xi^k=c_2+\sum_{k\in K}\sigma_{kr}\eta^k
\end{equation}
for some constants $c_1,c_2\in\bC$. In other words, we estimate the probability of the event that
\begin{equation}
  \label{eq:xieta}
c_1+\sum_{k\in K'}\xi^k=c_2+\sum_{k\in K'}\eta^k
\end{equation}
for a random subset $K'\subseteq K$. We show that this equality can be true for at most one~$K'$.

Indeed, assume that Equality \refeq{xieta} holds true for two different subsets $K'=K_1$ and $K'=K_2$ of $K$.
This implies that $\sum_{k\in K'}(\xi^k-\eta^k)=c_2-c_1$ for both $K'=K_1$ and $K'=K_2$ and, therefore,
\begin{equation}
  \label{eq:KK}
 \sum_{k\in K_1}\xi^k+\sum_{k\in K_2}\eta^k=\sum_{k\in K_2}\xi^k+\sum_{k\in K_1}\eta^k.
\end{equation}
By Conditions (i)--(ii), both sides of \refeq{KK} are sums of $|K_1|+|K_2|$ distinct numbers.
Since $K_1\ne K_2$, we can without loss of generality suppose that $K_1\not\subseteq K_2$.
Taking $k\in K_1\setminus K_2$, we see that the number $\xi^k$ occurs only in the left hand side of \refeq{KK}.
Thus, we have equality of two sums of $\zeta_n^j$ over different sets of indices $j$.
By multiplying both sides of Equation \refeq{KK} by a suitable $\zeta_n^t$, and using Condition (iii), we
can ``rotate'' these index sets modulo $n$ so that they lie within the interval $[1,n/\ln n]$.
This leads to a contradiction with Lemma~\ref{lem:cyclotomic}.

We conclude that Equality \refeq{xieta}, and therefore also Equality \refeq{k=m}, holds with probability at most
$2^{-|K|}=n^{-\Omega(\ln n)}$.

The third subcase of Case A, which relies on the same idea and is somewhat simpler,
is omitted from this outline (it appears as Case 3 in the proof of Claim \ref{cl:case1} below).

\Case B{$g$ is large, and $g'$ is small.}%
Since $\xi=\zeta_g=e^{-2\pi\mathrm{i}/g}$, the first $m$ elements
of the sequence \refeq{xi} are contained in an arc of $\bU$ of length $2\pi m/g$. The corresponding
elements of the sequence \refeq{eta} can take on at most $g'$ different values.  Therefore, there is
a set $K\subseteq\{1,\ldots,m\}$ of size
$$
|K|\ge m/g'
$$
such that if $k\in K$, then $\eta^k=\eta'$
for the same $\eta'$. Expose all $\sigma_j$ excepting $\sigma_{kr}$ for $k\in K$.
Similarly to Case A, Equality \refeq{k=m} converts into Equality \refeq{c1c2}
for some constants $c_1$ and $c_2$. This event can be recast as
\begin{equation}
  \label{eq:xietaB}
  c_1+\sum_{k\in K'}\xi^k=c_2+|K'|\eta'
\end{equation}
for a random subset $K'\subseteq K$. By Chernoff's bound, the size of $K'$ is concentrated
in the interval $\frac12|K|-|K|^{2/3}<|K'|<\frac12|K|+|K|^{2/3}$ with probability no less
than $1-2\exp(-|K|^{1/3})$. Fix an integer $m'$ such that
$\frac12|K|-|K|^{2/3}<m'<\frac12|K|+|K|^{2/3}$ and consider the event \refeq{xietaB}
conditioned on $|K'|=m'$. Under this condition, Equality \refeq{xietaB} reads
\begin{equation}
  \label{eq:c2metac1}
\sum_{k\in K'}\xi^k=c_2+m'\eta'-c_1.
\end{equation}
We choose $m$ such that
\begin{equation}
  \label{eq:m-large}
2\pi m/g<1/\ln n  
\end{equation}
while $m/g'$, and hence $|K|$, is large.
Condition \refeq{m-large} ensures that Lemma~\ref{lem:cyclotomic} applies, implying that
Equality \refeq{c2metac1} holds for at most one subset $K'\subseteq K$.
It follows that this equality holds with probability at most $1/{|K|\choose m'}$.
We therefore conclude that Equality \refeq{k=m} holds with probability at most
\begin{equation}
  \label{eq:bin-exp}
1/{|K|\choose |K|/2+|K|^{2/3}}+2\exp\of{-|K|^{1/3}}.  
\end{equation}
The specification of ``large'' $g$ and ``small'' $g'$ in the formal argument below ensures
that this probability is subexponentially small.

\Case C{$g$ is small.}%
This case can be treated without invoking Lemma~\ref{lem:cyclotomic}.
Let $J=\Set{j<n}{\zeta_n^{aj}=\xi}$, and note that $|J|=n/g$.
The set $J$ contains a subset $J'$ of size $|J'|\ge|J|/h\ge n/g^2$
such that the values $\zeta_n^{bj}$ are all equal for $j\in J'$.
That is, for all $j\in J'$, we have $\zeta_n^{aj}=\xi$ and $\zeta_n^{bj}=\eta'$
where $\eta'$ is the same $h$-th root of unity. Moreover, $J'\subseteq J$ can be chosen so that
$\xi\ne\eta'$; see the proof of Claim \ref{cl:case3} below for details.

Let us expose all random variables $\sigma_j$ except those for $j\in J'$.
Equality \refeq{k=m} then implies that
$$
\sum_{j\in J'}\sigma_j(\xi-\eta')=c
$$
for a constant $c\in\bC$ and, therefore, the sum $\sum_{j\in J'}\sigma_j$ evaluates to a constant value.
The probability of the last event is bounded by
${|J'|\choose\lfloor|J'|/2\rfloor}2^{-|J'|} \le |J'|^{-1/2} \le gn^{-1/2}$.
Although this bound is weaker than those obtained in Cases A and B,
it is still sufficient for use with the union bound \refeq{union-bound}
because, as we will see, the number of pairs $a,b$ covered by Case C is relatively small.

\bigskip

After this outline, we proceed to the detailed proof.
We begin by noting that the numbers $\xi$ and $\eta$, defined by \refeq{xi-eta-def},
are distinct. Indeed, the equality $\xi=\eta$ would imply $h=g$.
By \refeq{znzg}, this leads to
$$
\zeta_g^{a'r}=\zeta_n^{ar}=\zeta_n^{br}=\zeta_h^{b'r}=\zeta_g^{b'r},
$$
which holds only if $r(a'-b')$ is divisible by $g$. Since $r$ and $g$ are coprime,
we conclude that $a'-b'$ must be divisible by $g$, which is possible only when $a'=b'$,
contradicting the assumption that $a\ne b$.

Let $\varepsilon>0$ be an arbitrarily small constant.
Once this parameter is fixed, we assume that $n$ is sufficiently large.
We divide the remainder of the proof into Claims \ref{cl:case1}, \ref{cl:case2}, and \ref{cl:case3},
corresponding to Cases A, B, and C discussed above.
We present a detailed argument for digraphs (Part 1 of the theorem),
which, with minor modifications, also applies to graphs (Part 2).
We comment on these modifications at the end of the proof.

For the argument $\arg(z)$ of a complex number $z$, we suppose that $\arg(z)\in[0,2\pi)$.

\begin{claim}
Let $P_1=\Set{(a,b)}{g\geq n^{6\varepsilon}\text{ and }g'\geq n^{\varepsilon}}$. Then
$$
\sum_{(a,b)\in P_1}\prob{\lambda_a=\lambda_b} \le n^{-0.5\ln n}.
$$
\label{cl:case1}
\end{claim}

\begin{proof}
We claim that there exists $s$ such that
\begin{equation}
  \label{eq:1slnn}
1\leq s\leq \lceil\ln^4n\rceil  
\end{equation}
and
$$
\text{either }0<\arg(\eta^s)\le2\pi/\ln^4n \text{ or } 0<\arg(\eta^{-s})\le2\pi/\ln^4n.
$$
Indeed, consider the set $\mathcal{S}=\Set{\eta^s}{1\leq s\leq\lceil\ln^4n\rceil}$.
All elements of $\mathcal{S}$ are pairwise distinct. Indeed, suppose that $\eta^{s_1}=\eta^{s_2}$
for $1\le s_1<s_2\le\lceil\ln^4n\rceil$. Then $\eta^{s_2-s_1}=1$, so $s_2-s_1$ must be a multiple of the order
of $\eta$, which is $g'$. This yields a contradiction because, since $n$ is assumed to be sufficiently large,
we have
$$
s_2-s_1<\ln^4n<n^{\varepsilon}\le g'.
$$
Let $\eta^{s_1}$ and $\eta^{s_2}$, with $s_1<s_2$, be two elements of $\mathcal{S}$
such that the distance between them on the unit circle $\bU$ is minimal among all such pairs.
Since this distance is at most $2\pi/\ln^4n$, we can set $s=s_2-s_1$.

We now consider three cases and show that, in each of them,
\begin{equation}
  \label{eq:n-lnn}
\prob{\lambda_a=\lambda_b} \le n^{-0.6\ln n}.  
\end{equation}

\Case1{$\arg(\eta^{s})\le2\pi/\ln^4n$.}%
We follow the strategy presented in Case A of the outline above.
To establish the upper bound \refeq{n-lnn}, it suffices to find a set $K\subseteq\{1,\ldots,n-1\}$
of size $|K|=\lceil\ln^2n\rceil$ that satisfies Conditions (i)--(iii) stated above.
We set $K=\Set{s,2s,\ldots,\lceil\ln^2 n\rceil s}$.

Condition (i) for this set is ensured by the estimate $\arg(\xi)=2\pi-2\pi/g \ge 2\pi-2\pi/n^{6\varepsilon}$
and by the definition of Case 1. Since $n$ is large enough, due to \refeq{1slnn} we have $\arg(\xi^{\ell s})>2\pi-2\pi/n^{5\varepsilon}$
for all $\ell\leq\lceil\ln^2n\rceil$. In the case under consideration we also have $\arg(\eta^{\ell s})<\pi/\ln n$
for all $\ell\le\lceil\ln^2 n\rceil$. Consequently,
$$
\Set{\arg(\xi^k)}_{k\in K}\subset(2\pi-\pi/\ln n,2\pi)\text{ and }\Set{\arg(\eta^k)}_{k\in K}\subset(0,\pi/\ln n). 
$$
This implies the other two Conditions (ii) and~(iii), yielding the upper bound
\begin{equation}
  \label{eq:Pab}
\prob{\lambda_a=\lambda_b} \le 2^{-|K|}\le n^{-0.6\ln n}.  
\end{equation}

\Case2{$\arg(\eta^{-s})\le2\pi/\ln^4n$ and $\eta^s\neq\xi^s$.}%
The first of the two assumptions
ensures---similarly to Case 1---that Condition (i) is fulfilled for every set
$K\subseteq\Set{s,2s,\ldots,(2\lceil\ln^2 n\rceil) s}$. Each such set $K$ also satisfies
Condition (iii) because, similarly to Case 1, we have
$$
\Set{\arg(\xi^k)}_{k\in K}\subset(2\pi-2\pi/\ln n,2\pi)\text{ and }\Set{\arg(\eta^k)}_{k\in K}\subset(2\pi-2\pi/\ln n,2\pi).
$$
These inclusions, however, do not guarantee Condition (ii). Nevertheless, we can show that
Condition (ii) holds for at least one set $K\subseteq\{s,2s,\ldots,(2\lceil\ln^2 n\rceil) s\}$
such that $|K|\ge\lceil\ln^2 n\rceil$.
 
Indeed, assume without loss of generality that $\arg(\xi^{-s})<\arg(\eta^{-s})$.
The existence of a suitable set $K$ follows from the observation that, among any two consecutive values
$\arg(\xi^{\ell s})$ and $\arg(\xi^{(\ell+1)s})$ for $\ell<2\lceil\ln^2 n\rceil$, at most one can belong to the set
$\Set{\arg(\eta^k)}_{k\in \{s,2s,\ldots,(2\lceil\ln^2 n\rceil) s\}}$, because the distance between any two
neighboring elements of this set on the unit circle $\bU$ (which is equal to $\arg(\eta^{-s})$)
is larger than the distance between $\xi^{\ell s}$ and $\xi^{(\ell+1)s}$ (which is equal to $\arg(\xi^{-s})$).
It follows that the set $K$, consisting of all those $\ell s$ with $\ell\le2\lceil\ln^2 n\rceil$
such that $\arg(\xi^{\ell s})\notin\Set{\arg(\eta^k)}_{k\in \{s,2s,\ldots,(2\lceil\ln^2 n\rceil) s\}}$,
has the desired size. Condition (ii) is satisfied for this set by construction,
and we obtain the upper bound \refeq{Pab} also in this case.

\Case3{$\arg(\eta^{-s})\le2\pi/\ln^4n$ and $\eta^s=\xi^s$.}%
Let $K=\Set{0,s,2s,\ldots,\lceil\ln^2 n\rceil s}$.
Following our general strategy, let us expose all random variables $\sigma_j$
excepting $\sigma_{(k+1)r}$ for $k\in K$. The event \refeq{k=m} can now be recast as the equality
\begin{equation}
  \label{eq:xik1-etak1}
\sum_{k\in K'}(\xi^{k+1}-\eta^{k+1})=c  
\end{equation}
for a constant $c\in\bC$ and a random subset $K'\subseteq K$.
For $k=\ell s$, we have
$$
\xi^{k+1}-\eta^{k+1}=\xi^{\ell s+1}-\eta^{\ell s+1}=\xi^{\ell s+1}-\xi^{\ell s}\eta=\xi^{\ell s}(\xi-\eta),
$$
which allows us to rewrite Equality \refeq{xik1-etak1} as
\begin{equation}
  \label{eq:sumxik}
\sum_{k\in K'}\xi^k=c/(\xi-\eta) 
\end{equation}
(recall that $\xi\ne\eta$). Since $\arg(\xi^s)=\arg(\eta^s)\ge2\pi-2\pi/\ln^4n$,
we have $\arg(\xi^{\ell s})>2\pi-2\pi/\ln n$ for all $\ell\le\lceil\ln^2 n\rceil$.
We can, therefore, use Lemma \ref{lem:cyclotomic}
to conclude that Equality \refeq{sumxik} can be true for at most one $K'\subseteq K$.
It follows that Equality \refeq{xik1-etak1}, and as well Equality \refeq{k=m},
holds with probability at most $2^{-|K|} \le n^{-0.6\ln n}$.

\medskip

Thus, Bound \refeq{n-lnn} is established in each of the three cases. Consequently,
$$
\sum_{(a,b)\in P_1}\prob{\lambda_a=\lambda_b} \le |P_1|\cdot n^{-0.6\ln n} \le n^{-0.5\ln n},
$$
where the last inequality holds because $n$ is assumed to be sufficiently large.
The proof of the claim is complete.
\end{proof}

\begin{claim}
Let $P_2=\Set{(a,b)}{g\geq n^{6\varepsilon}\text{ while }g'<n^{\varepsilon}}$. Then
$$
\sum_{(a,b)\in P_2}\prob{\lambda_a=\lambda_b}\le\exp\of{-n^\varepsilon}.
$$
\label{cl:case2}
\end{claim}

\begin{proof}
We closely follow the proof strategy presented (in a fairly complete form) in Case B of the outline above.
Specifically, we set $m=n^{5\varepsilon}$.
Since $n$ is assumed to be sufficiently large, this implies $|K|\ge m/g'>n^{4\varepsilon}$
and also ensures Bound \refeq{m-large}. This justifies the probability bound \refeq{bin-exp}.
Using the estimate ${n\choose k}\ge(n/k)^k$, we conclude from \refeq{bin-exp} that if $(a,b)\in P_2$, then
$$
\prob{\lambda_a=\lambda_b}
<\of{\frac12+|K|^{-1/3}}^{|K|/2}+2\exp\of{-|K|^{1/3}}.
$$
Since $n$ is sufficiently large, this probability is bounded by $3\exp\of{-n^{4\varepsilon/3}}$
and, therefore,
$$
\sum_{(a,b)\in P_2}\prob{\lambda_a=\lambda_b}\le|P_2|\cdot3\exp\of{-n^{4\varepsilon/3}}\le\exp\of{-n^{\varepsilon}},
$$
as required.
\end{proof}

\begin{claim}
Let $P_3=\Set{(a,b)}{g < n^{6\varepsilon}}$. Then
$$
\sum_{(a,b)\in P_3}\prob{\lambda_a=\lambda_b} \le n^{-1/2+20\varepsilon}.
$$
\label{cl:case3}
\end{claim}

\begin{proof}
As explained in Case C of the outline, we begin by considering the set $J=\Set{j<n}{\zeta_n^{aj}=\xi}$,
which has size $|J|=n/g>n^{1-6\varepsilon}$. Since $\zeta_n^{bj}$ can take on at most $h\le g$
distinct values, $J$ contains a subset $J'\subseteq J$ of size $|J'|\ge|J|/g>n^{1-12\varepsilon}$
such that the values $\zeta_n^{bj}$ for $j\in J'$ are all equal to the same $h$-th
root of unity $\eta'$. If $h<g$, then clearly $\eta'\ne\xi$.
If $h=g$, then we define $J'$ explicitly by $J'=\Set{\ell g+1}{\ell=0,1,\ldots,\lfloor(n-2)/g\rfloor}$,
which yields $\eta'=\eta$ and, therefore, $\eta'\ne\xi$ in any case.

Exposing all $\sigma_j$ except those for $j\in J'$, we obtain from Equality \refeq{k=m}
that $\sum_{j\in J'}\sigma_j=c$ for a constant $c$. This occurs with probability at most
$$
{|J'|\choose\lfloor|J'|/2\rfloor}2^{-|J'|} < |J'|^{-1/2} < n^{-1/2+6\varepsilon},
$$
providing us with an upper bound for $\prob{\lambda_a=\lambda_b}$ if $(a,b)\in P_3$.

We now estimate the number of pairs $(a,b)$ in $P_3$.
Recall that $a=\gcd{(a,n)}\cdot a'$, where $a'\le n/\gcd{(a,n)}=g<n^{6\varepsilon}$.
The factor $\gcd{(a,n)}$ of $a$ can take on at most $d(n)$ different values,
where $d(n)$ denotes the total number of divisors of $n$.
It is known \cite[Theorem 13.12]{Apostol98} that $d(n)=n^{O(1/\ln\ln n)}$.
Since there are less than $n^{6\varepsilon}$ possibilities to choose the factor $a'$,
the element $a$ can be chosen in at most $n^{7\varepsilon}$ ways, and the same
holds true as well for $b$. It follows that $|P_3|\le n^{14\varepsilon}$,
and we conclude that
$$
\sum_{(a,b)\in P_3}\prob{\lambda_a=\lambda_b} < |P_3|\cdot n^{-1/2+6\varepsilon} \le n^{-1/2+20\varepsilon},
$$
completing the proof of the claim.
\end{proof}

Claims~\ref{cl:case1},~\ref{cl:case2},~and~\ref{cl:case3} readily imply that
$$
\sum_{0\le a,b\le n-1}\prob{\lambda_a=\lambda_b} \le n^{-1/2+\varepsilon}
$$
for each $\varepsilon>0$ and sufficiently large $n$. This completes the proof of Part 1 of Theorem~\ref{thm:whp}.
The proof of Part 2
proceeds in essentially the same way. The only difference is that, for an unexposed
random variable $\sigma_j$, instead of the term $\sigma_j\zeta_n^{aj}$ we now
deal with $\sigma_j\zeta_n^{aj}+\sigma_{n-j}\zeta_n^{a(n-j)}=\sigma_j(\zeta_n^{aj}+\zeta_n^{-aj})$.
Lemma~\ref{lem:cyclotomic} remains applicable after multiplying the entire sum by an additional
factor of $\zeta_n^{t'}$ for a small value of $t'$. Finally,
in Case 1 of the proof of Claim A, one has to address also the possibility that $\eta^{-s}=\xi^s$,
which is handled similarly to Case 3 of the proof.

\begin{remark}
The probability bound in Theorem \ref{thm:whp} is nearly optimal.
This can be shown by noticing that a random digraph $\cay(\bZ_n,S)$ for $n=3p$ with $p$ prime
has repeated eigenvalues with probability $\Omega(n^{-1/2})$. Indeed, for $r=0,1,2$, let
$S_r=\Set{s\in S}{s=r\pmod3}$. Note that
$$
\lambda_p=\sum_{j\in S}\zeta_n^{pj}=\sum_{j\in S}\zeta_3^{j}=\sum_{j\in S_0}1+\sum_{j\in S_1}\zeta_3+\sum_{j\in S_2}\zeta_3^{2}.
$$
Likewise, 
$$
\lambda_{2p}=\sum_{j\in S}\zeta_n^{2pj}=\sum_{j\in S}\zeta_3^{2j}=\sum_{j\in S_0}1+\sum_{j\in S_1}\zeta_3^2+\sum_{j\in S_2}\zeta_3.
$$
As a consequence, $\lambda_p=\lambda_{2p}$ whenever $|S_1|=|S_2|$. The last equality is true with probability
$(1-o(1))/\sqrt{\pi n/3}$ because $|S_1|$ and $|S_2|$ are independent random variables with probability
distribution~$\mathrm{Bin}(p,1/2)$.

It can be similarly shown that, for $n=5p$, the spectrum of a random graph $\cay(\bZ_n,S)$
is not saturated with the same probability bound~$\Omega(n^{-1/2})$.
\end{remark}

\section{Proof of Theorem \ref{thm:main}}\label{s:main}

We are now ready to prove our main result. Specifically, we prove Theorem \ref{thm:main}
for each of the three concepts of a circulant:
\begin{itemize}
\item
  A Cayley (di)graph $X=\cay(\bZ_n,S)$. The uniform probability distribution of $X$ means
  that the connection set $S$ is equiprobably chosen among all subsets of $\bZ_n\setminus\{0\}$
  in the case of digraphs and among all inverse-closed subsets in the case of graphs.
\item
  An unlabeled circulant, i.e., an isomorphism class of Cayley (di)graphs $X=\cay(\bZ_n,S)$.
  The uniform distribution means that each isomorphism class on $\bZ_n$ is chosen equiprobably.
  In the algorithmic setting, an isomorphism class is presented by its representative
(a (di)graph from the class). Alternatively, we can think of the probability distribution
on all Cayley (di)graphs $X=\cay(\bZ_n,S)$ in which each $X$ appears with probability
$1/(\ell_n\, s(X))$, where $\ell_n$ is the total number of $n$-vertex unlabeled circulants
and $s(X)$ is the number of connection sets $S$ such that $\cay(\bZ_n,S)$ is isomorphic to~$X$.
\item
  A labeled circulant, i.e., an arbitrary (di)graph on the vertex set $\{0,1,\ldots,n-1\}$ isomorphic
  to some Cayley (di)graph $X=\cay(\bZ_n,S)$. The uniform distribution is considered on all $n$-vertex (di)graphs
  in this class.
\end{itemize}

In each of the three cases, we use the same canonization algorithm presented in Subsection \ref{ss:canon}.
For digraphs, the algorithm is extremely simple: We just individualize one vertex in an input digraph $X$
and run CR on the obtained vertex-colored graph $X_0$. In this way either we get an individual
label for each vertex of $X$ or the algorithm gives up. The labeling is canonical for all
circulants $X$, and it is successfully produced whenever $X_0$ is walk-discrete.
For graphs, the algorithm is a little bit more complicated and is discussed in detail in Subsection \ref{ss:canon}.
It succeeds whenever $X_0$ is walk-saturated.

Lemmas \ref{lem:simple} and \ref{lem:saturated} provide us with two sufficient spectral conditions:
$X_0$ is walk-discrete whenever $X$ has simple spectrum, and $X_0$ is walk-saturated 
whenever $X$ has saturated spectrum. This reduces our task to estimating the probability
that the random digraph $X$ has simple spectrum and, respectively, that the random
graph $X$ has saturated spectrum. In the case of Cayley (di)graphs, the proof is
completed by applying Theorem~\ref{thm:whp}.

It remains to show that the estimate of Theorem~\ref{thm:whp} stays as well true for
the uniformly distributed labeled and unlabeled circulants. To this end, we present
a general way to convert an estimate for one distribution into an estimate for
another distribution with a small overhead cost.

\subsection{Formal framework}

We introduce the following notation simultaneously for graphs and digraphs.
Let \allc be the set of all Cayley (di)graphs $X$ of cyclic groups, that is, all (di)graphs $X=\cay(\bZ_n,S)$
for any $n$ and $S$. Recall that the notion of a Cayley (di)graph was formally defined in Section \ref{ss:canon}.
According to this definition, two graphs $X=\cay(\bZ_n,S)$ and $Y=\cay(\bZ_n,T)$ in \allc
are different exactly when $S\ne T$.

In general, speaking of a (di)graph $X$, we always suppose that the vertex set of $X$
is $\{0,1,\ldots,n-1\}$, where $n$ is the order of $X$.
For a set of (di)graphs $\Q$, by $\Q_n$ we denote the set of the (di)graphs in $\Q$
that have order $n$. By $\un \Q$ we denote the unlabeled version of $\Q$ where
isomorphic graphs are not distinguished. Formally, $\un\Q$ is the quotient set of $\Q$
by the isomorphism relation. In other words, $\un \Q$ consists of all unlabeled graphs
whose representatives appear in  $\Q$. 
Furthermore, $\la \Q$ is defined to be the closure of $\Q$ under isomorphism, that is, if $\Q$
contains a (di)graph $X$ of order $n$, then $\la \Q$ contains all graphs on
the vertex set $\{0,1,\ldots,n-1\}$ isomorphic to $X$.
We write $\un \Q_n$ and $\la \Q_n$ to denote the restrictions of $\un \Q$ and $\la \Q$
to the (di)graphs of order~$n$.

Note that $\un \allc$ is precisely the set of unlabeled circulants, and
$\la \allc$ is the set of labeled circulants.

For an arbitrary set of (di)graphs $\Q$, let $a(\Q_n)$ denote the minimum number of automorphisms of a (di)graph in $\Q_n$.
Note that
\begin{equation}\label{eq:QuQ}
 |\la\Q_n|\le|\un\Q_n|\,n!/a(\Q_n). 
\end{equation}

The following important fact is a consequence of the main result in \cite{Muzychuk04};
see \cite[Theorem 1.1]{Muzychuk04} and the discussion right after its statement.

\begin{proposition}[{Muzychuk \cite{Muzychuk04}}]\label{prop:muzychuk}
  For every $S\subseteq \bZ_n\setminus\{0\}$ there are at most $\phi(n)$ sets
$S'\subseteq \bZ_n\setminus\{0\}$ such that $\cay(\bZ_n,S')\cong\cay(\bZ_n,S)$.
\end{proposition}

Proposition \ref{prop:muzychuk} readily implies that if $\Q\subseteq\allc$, then
\begin{equation}\label{eq:QlQ}
|\un \Q_n|\le|\Q_n|\le\phi(n)|\un\Q_n|.
\end{equation}

Seeing $\bZ_n$ as a ring with addition and multiplication modulo $n$,
we write $\bZ_n^\times$ to denote the multiplicative group of order $\phi(n)$ consisting
of the invertible elements of~$\bZ_n$.
For a set $S\subseteq\bZ_n\setminus\{0\}$, we define the subgroup $K(S)\le\bZ_n^\times$
by $K(S)=\Set{k\in \bZ_n^\times}{kS=S}$. In the case of digraphs,
we call a connection set $S$ \emph{multiplier-free} if $K(S)=\{1\}$. In the case of graphs,
an inverse-closed connection set $S=-S$ is called \emph{multiplier-free} if $K(S)=\{1,-1\}$.
The set of Cayley (di)graphs with multiplier-free connection sets is denoted by~\asym.

We say that a set $\Q\subseteq\allc$ is \emph{isomorphism-invariant within $\allc$}
if for every $X\in\Q$ and $Y\in\allc$, we have $Y\in\Q$ whenever $X\cong Y$.

\begin{lemma}\label{lem:zibin}
  The set \asym is isomorphism-invariant within~$\allc$.
\end{lemma}

\begin{proof}
  Assume that
  \begin{equation}
    \label{eq:ScongT}
\cay(\bZ_n,S)\cong\cay(\bZ_n,T).    
  \end{equation}
  Let $kS=S$ for $k\in \bZ_n^\times$.
  To obtain the lemma, it is enough to prove that $kT=T$ as well.

  We use the following fact, which was conjectured by Zibin and proved by
  Muzychuk, Klin, and Pöschel \cite{MuzychukKP99}. Let $d\mid n$, i.e., $d$ is a divisor of $n$.
  For $S\subseteq\bZ_n$, define $(S)_d=\Set{s\in S}{\gcd(s,n)=d}$. Then, according to \cite[Theorem 5.1]{MuzychukKP99},
  the isomorphism \refeq{ScongT} implies that for every $d\mid n$ there exists $m_d\in \bZ_n^\times$ such that
  $(T)_d=m_d(S)_d$.

  For $k\in \bZ_n^\times$, let $\mu_k$ be the permutation of $\bZ_n$ defined by $\mu_k(x)=kx$.
  The equality $kS=S$ means that $S$ is a union of orbits of $\mu_k$, that is,
  there is a set $H\subset\bZ_n$ such that
  $$
S=\bigcup_{h\in H}\langle k\rangle h=\bigcup_{d\mid n}\bigcup_{h\in (H)_d}\langle k\rangle h,
$$
where $\langle k\rangle$ denotes the subgroup of $\bZ_n^\times$ generated by $k$.
Note that $(S)_d=\bigcup_{h\in (H)_d}\langle k\rangle h$. It follows that
$$
T=\bigcup_{d\mid n}(T)_d=\bigcup_{d\mid n}m_d(S)_d=\bigcup_{d\mid n}\bigcup_{h\in (H)_d}\langle k\rangle (hm_d).
$$
This shows that $T$ is also a union of orbits of $\mu_k$ and, as required, $kT=T$.  
\end{proof}

If $\Q\subseteq\asym$ is isomorphism-invariant within $\allc$, then the lower bound in \refeq{QlQ} can be improved as follows:
\begin{eqnarray}
&|\Q_n|&=\ \phi(n)|\un\Q_n|\quad\text{for digraphs,}\label{eq:QuQdi}\\
\frac{\phi(n)}2\,|\un\Q_n|\ \le&|\Q_n|&\le\ \phi(n)|\un\Q_n|\quad\text{for graphs.}\label{eq:QuQg}
\end{eqnarray}
Indeed, if $\cay(\bZ_n,S)\in\asym$ and $k,k'\in \bZ_n^\times$, then the isomorphic copies
$\cay(\bZ_n,kS)$ and $\cay(\bZ_n,k'S)$ of this graph are distinct (i.e., $kS\ne k'S$)
whenever $k'\ne k$ in the case of digraphs and $k'\ne \pm k$ in the case of graphs.

\begin{lemma}\label{lem:m-free}\hfill
  \begin{enumerate}[\bf 1]
  \item 
  A random connection set $S\subseteq\bZ_n\setminus\{0\}$ is not multiplier-free
  with probability less than $n\,2^{-n/4}$.
  \item 
  A random inverse-closed connection set $S=-S$ is not multiplier-free
  with probability less than $2n\,2^{-n/8}$.
  \end{enumerate}  
\end{lemma}

\begin{proof} 
Define the \emph{annihilator} of $a\in\bZ_n$ by $\ann(a)=\Set{x\in\bZ_n}{xa=0}$.
Since $\ann(a)$ is a subgroup of $\bZ_n$, we have $|\ann(a)|\le n/2$ for every~$a\ne0$.

1.
It suffices to prove that, for each $a\ne1$ in $\bZ_n^\times$, the equality $aS=S$
is fulfilled with probability at most $2^{-n/4}$. Let $\mu_a$ be the permutation
of $\bZ_n$ defined by $\mu_a(x)=ax$. We have $\mu_a(x)=x$ exactly when $x\in\ann(a-1)$.
Thus, $\mu_a$ is the identity on $\ann(a-1)$ and a fixed-point-free permutation of
the set $\bZ_n\setminus\ann(a-1)$.
Denote the restriction of $\mu_a$ to this set by $\mu'_a$.
Let $c_1,\ldots,c_t$ be the cycle type of $\mu'_a$. Note that $\sum_{i=1}^tc_i=n-|\ann(a-1)|\ge n/2$.
Note also that $t\le(\sum_{i=1}^tc_i)/2$ because $c_i\ge 2$ for all $i$.
The equality $aS=S$ is true if and only if every cycle of $\mu'_a$ either is entirely
in $S$ or is disjoint from $S$. This happens with probability
$$
\prod_{i=1}^t 2^{-c_i+1}=2^{-(\sum_{i=1}^tc_i)+t}\le2^{-(\sum_{i=1}^tc_i)/2}\le2^{-n/4}.
$$
2.
Again, it is enough to prove that, for each $a\ne\pm1$ in $\bZ_n^\times$, the equality $aS=S$
is fulfilled with probability at most $2^{-n/8+2}$. Let $Z$ be the set of all pairs
$\{x,-x\}$ for $x\in\bZ_n$ such that $x\ne-x$. Note that $|Z|=\lfloor(n-1)/2\rfloor$.
The permutation $\mu_a$ naturally acts on $Z$ by $\mu_a(\{x,-x\})=\{ax,-ax\}$.
Let us estimate the number of fixed points under this action. We have $\mu_a(\{x,-x\})=\{x,-x\}$
if and only if $ax=x$ or $ax=-x$, which happens exactly when $x\in\ann(a-1)\cup\ann(a+1)$.
Note that
\begin{equation}
  \label{eq:annaa}
|\ann(a-1)\cup\ann(a+1)|\le \frac n2+1.
\end{equation}
Indeed, let $b=\gcd(a-1,n)$ and $c=\gcd(a+1,n)$, and note that $|\ann(a-1)|=b$ and $|\ann(a+1)|=c$.
Since the sets $\ann(a-1)$ and $\ann(a+1)$ share at least one element, namely $0$,
it is enough to prove that $b+c\le n/2+2$. Since neither $a-1=0$ nor $a+1=0$, neither $b$ nor $c$
exceeds $n/2$. Hence, we are immediately done if $b\le2$ or $c\le2$.
Suppose, therefore, that both $b\ge3$ and $c\ge3$.

Let $d=\gcd(b,c)$. Since $d$ divides both $a-1$ and $a+1$, we have $d\le2$.
If $d=1$, then $b$ and $c$ are coprime divisors of $n$ and, therefore,
$b+c\le b+n/b$. Note that $3\le b\le n/c\le n/3$. In particular, $n\ge9$ in this case.
It follows that $b+c\le n/3+3<n/2+2$.

If $d=2$, then $b/2$ and $c/2$ are coprime divisors of $n/2$. The similar argument
yields
$$
\frac b2 + \frac c2 \le \frac b2 + \frac{n/2}{b/2} \le \frac n3 + \frac32 < \frac n2+2,
$$
completing the proof of Bound~\refeq{annaa}.

Since $x$ belongs to $\ann(a-1)\cup\ann(a+1)$ simultaneously with $-x$,
the number of fixed points under the action of $\mu_a$ on $Z$ is at most $n/4$. It follows that
at least $n/4-1$ pairs in $Z$ are non-fixed.
Similarly to Part 1, we conclude that $aS=S$ with probability at most~$2^{-n/8+1}$.
\end{proof}

To conclude the notational stuff, let $\Q\subseteq\allc$.
Then $\prob{\Q_n}=|\Q_n|/|\allc_n|$ is the probability that a random Cayley (di)graph $\cay(\bZ_n,S)$,
where $S$ is chosen equiprobably among all possible connection sets $S\subseteq\bZ_n\setminus\{0\}$, belongs to $\Q$.
Exactly this random Cayley (di)graph model is considered in Theorem \ref{thm:whp}
and studied in Section \ref{s:whp}. In what follows, 
we also write $\probla{\la\Q_n}$ to denote the probability that a random labeled circulant
of order $n$, i.e., a (di)graph chosen randomly and uniformly in $\la\allc_n$, belongs to $\la\Q_n$. Thus,
$\probla{\la\Q_n}=|\la\Q_n|/|\la\allc_n|$. Similarly, $\probun{\un\Q_n}=|\un\Q_n|/|\un\allc_n|$
is the probability that a random unlabeled circulant of order $n$ belongs to~$\un\Q_n$.

\subsection{From Cayley (di)graphs to unlabeled circulants}\label{ss:cay-to-u}

\begin{lemma}\label{lem:u}
  Let $\R\subseteq\allc$. If $\R$ is isomorphism-invariant within $\allc$, then
\begin{eqnarray*}
  \probun{\un\R_n}&\le&(1+o(1))\,\prob{\R_n}+n^2\,2^{-n/4}\quad\text{for digraphs,}\label{eq:udi}\\
  \probun{\un\R_n}&\le&(2+o(1))\,\prob{\R_n}+2n^2\,2^{-n/8}\quad\text{for graphs.}\label{eq:ug}                     
\end{eqnarray*}  
\end{lemma}

\begin{proof}
We prove the inequality for graphs; the case of digraphs is similar.
Note that $\un\allc_n\setminus\un\asym_n=\un{(\allc_n\setminus\asym_n)}$ due to Lemma \ref{lem:zibin}.
Applying the inequalities \refeq{QlQ} to $\Q=\allc\setminus\asym$ and to $\Q=\allc$,
we derive from Lemma \ref{lem:m-free}
\begin{equation}\label{eq:m-free}
  \probun{\un\allc_n\setminus\un\asym_n}=\frac{|\un\allc_n\setminus\un\asym_n|}{|\un\allc_n|}
=\frac{|\un{(\allc_n\setminus\asym_n)}|}{|\un\allc_n|}\le
\frac{|\allc_n\setminus\asym_n|}{|\allc_n|/\phi(n)}=\phi(n)\,\prob{\allc_n\setminus\asym_n}\le 2n^2\,2^{-n/8}.
\end{equation}

Note that $\un\R_n\cap\un\asym_n=\un{(\R_n\cap\asym_n)}$ because
$\asym$ is isomorphism-invariant within $\allc$ by Lemma \ref{lem:zibin}
and $\R$ is isomorphism-invariant within $\allc$ by assumption.
Using Bound \refeq{m-free} and applying the inequalities \refeq{QuQg} to $\Q=\R\cap\asym$ and to $\Q=\asym$, we obtain
\begin{multline*}
\probun{\un\R_n}\le\probun{\un\R_n\cap\un\asym_n}+\probun{\un\allc_n\setminus\un\asym_n}\le
\frac{\probun{\un\R_n\cap\un\asym_n}}{\probun{\un\asym_n}}+2n^2\,2^{-n/8}
=\frac{|\un\R_n\cap\un\asym_n|}{|\un\asym_n|}+2n^2\,2^{-n/8}\\
=\frac{|\un{(\R_n\cap\asym_n)}|}{|\un\asym_n|}+2n^2\,2^{-n/8}
\le\frac{2\,|\R_n\cap\asym_n|/\phi(n)}{|\asym_n|/\phi(n)}+2n^2\,2^{-n/8}\le
\frac{2\,|\R_n|}{|\asym_n|}+2n^2\,2^{-n/8}\\
=\frac{2\,\prob{\R_n}}{\prob{\asym_n}}+2n^2\,2^{-n/8}
=(2+o(1))\,\prob{\R_n}+2n^2\,2^{-n/8}.
\end{multline*}
For the last equality we used Lemma \ref{lem:m-free} once again.
\end{proof}

Let $\R$ be the set of all Cayley digraphs $X=\cay(\bZ_n,S)$ whose spectrum is \emph{not} simple.
In the case of graphs, we set $\R$ to be the set of all $X=\cay(\bZ_n,S)$ whose spectrum is \emph{not} saturated.
Lemma \ref{lem:u} implies that the bound of Theorem \ref{thm:whp} holds true also for unlabeled circulants.
This completes the proof of Theorem \ref{thm:main} in the unlabeled case.

\subsection{From unlabeled to labeled circulants}\label{ss:firm}

For $a\in\bZ_n$, define a bijection $\sigma_a\function{\bZ_n}{\bZ_n}$ by $\sigma_a(x)=x+a$.
For every Cayley digraph $X=\cay(\bZ_n,S)$,
the map $\sigma_a$ is an automorphism of $X$ for all $a$.
If $X$ has no other automorphism, that is, $\aut(X)$ is as small as possible
(see \cite[Section 8.1]{DobsonMM22}), then we call $X$ \emph{firm}.
In other words, $X$ is firm exactly when $\aut(X)\cong\bZ_n$.

Now, let $X$ be a Cayley graph. In this case there is also another automorphism $\rho$
defined by $\rho(x)=-x$. The automorphisms $\sigma_1$ and $\rho$ generate a subgroup of $\aut(X)$
isomorphic to the dihedral group $D_{2n}$. If $X$ has no other automorphisms, i.e.,
$\aut(X)\cong D_{2n}$, then we say that $X$ is a \emph{firm Cayley graph} of $\bZ_n$.
Note that a firm graph is not firm as a digraph; this should not
make any confusion because we treat random graphs and random digraphs separately
(even when in parallel).

The set of firm (di)graphs is denoted by \firm.
The following equalities easily follow from the definitions:
If $\Q\subseteq\firm$, then
\begin{eqnarray}
|\la\Q_n|&=&(n-1)!\,|\un\Q_n|\quad\text{for digraphs,}\label{eq:firm-l-di}\\
|\la \Q_n|&=&\frac{(n-1)!}2\,|\un\Q_n|\quad\text{for graphs.}\label{eq:firm-l-g}                     
\end{eqnarray}  

We will need the following estimates obtained in \cite{BhoumikDM14,DobsonSV16}.\footnote{%
Note that, since $\firm\subseteq\asym$, we could use Proposition~\ref{prop:firm} instead of Lemma~\ref{lem:m-free}
in Subsection~\ref{ss:cay-to-u}. However, Lemma~\ref{lem:m-free} has the advantage of being proved by an elementary method.
}

\begin{proposition}\label{prop:firm}\hfill
  \begin{enumerate}[\bf 1.]
  \item
    $\prob{\allc_n\setminus\firm_n}=2^{-n/4+o(n)}$ for digraphs
    \textup{(Dobson, Spiga, and Verret \cite[Theorem~1.6]{DobsonSV16}).}
  \item
    $\prob{\allc_n\setminus\firm_n}=O(n2^{-n/8})$ for graphs
    \textup{(Bhoumik, Dobson, and Morris \cite[Theorem~3.2]{BhoumikDM14}).}  
  \end{enumerate}
\end{proposition}

Note that $\un\allc_n\setminus\un\firm_n=\un{(\allc_n\setminus\firm_n)}$
because $\firm$ is obviously isomorphism-invariant within $\allc$.
Applying the inequalities \refeq{QlQ} to $\Q=\allc\setminus\firm$ and to $\Q=\allc$,
similarly to \refeq{m-free} we get the relation
$$
\probun{\un\allc_n\setminus\un\firm_n}\le\phi(n)\,\prob{\allc_n\setminus\firm_n}.
$$
By Proposition \ref{prop:firm}, this implies that
\begin{eqnarray}
  \probun{\un\allc_n\setminus\un\firm_n}&=&2^{-n/4+o(n)}\quad\text{for digraphs,}\label{eq:firm-di}\\
  \probun{\un\allc_n\setminus\un\firm_n}&=&O(n^2\,2^{-n/8})\quad\text{for graphs.}\label{eq:firm-g}                     
\end{eqnarray}

\begin{lemma}\label{lem:l}
  Let $\cS\subseteq\allc$. If $\cS$ is isomorphism-invariant within $\allc$, then
\begin{eqnarray*}
  \probla{\la\cS_n}&\ge&\probun{\un\cS_n}-2^{-n/4+o(n)}\quad\text{for digraphs,}\\
  \probla{\la\cS_n}&\ge&\probun{\un\cS_n}-O(n^2\,2^{-n/8})\quad\text{for graphs.}                    
\end{eqnarray*}  
\end{lemma}

\begin{proof}
We prove the inequality for graphs; the case of digraphs is similar.
Note that $\un\cS_n\cap\un\firm_n=\un{(\cS_n\cap\firm_n)}$ and $\la\cS_n\cap\la\firm_n=\la{(\cS_n\cap\firm_n)}$
because both $\cS$ and $\firm$ are isomorphism-invariant within $\allc$.
Applying the inequality \refeq{QuQ} to $Q=\allc$ and the equality \refeq{firm-l-g}
to $\Q=\cS\cap\firm$, we derive
\begin{multline}
\probla{\la\cS_n}\ge\probla{\la\cS_n\cap\la\firm_n}=
\frac{|\la\cS_n\cap\la\firm_n|}{|\la\allc_n|}=\frac{|\la{(\cS_n\cap\firm_n)}|}{|\la\allc_n|}\\
\ge
\frac{(n-1)!\,|\un{(\cS_n\cap\firm_n)}|/2}{(n-1)!\,|\un\allc_n|/2}
=\frac{|\un\cS_n\cap\un\firm_n|}{|\un\allc_n|}=
\probun{\un\cS_n\cap\un\firm_n}.\label{eq:la-un}
\end{multline}
This implies that
$$
\probla{\la\cS_n}\ge\probun{\un\cS_n}-\probun{\un\allc_n\setminus\un\firm_n}\ge\probun{\un\cS_n}-O(n^2\,2^{-n/8}).
$$
The last inequality follows from the bound~\refeq{firm-g}.
\end{proof}

Now, let $\cS$ be the set of all Cayley digraphs $X=\cay(\bZ_n,S)$ with simple spectrum.
In the case of graphs, we set $\cS$ to be the set of all $X=\cay(\bZ_n,S)$ with saturated spectrum.
We already know that the bound of Theorem \ref{thm:whp} holds true for unlabeled circulants.
Lemma \ref{lem:l} implies that it is as well true for labeled circulants.

Lemma \ref{lem:l} provides a rather general way of showing that if a property holds for
almost all unlabeled circulants, then it also holds for almost all labeled circulants.
We remark that that for the property $\cS$ that a circulant digraph has simple spectrum
(or that a circulant graph has saturated spectrum) this can be alternatively derived from
Lemmas \ref{lem:simple}, \ref{lem:saturated}, and \ref{lem:w-saturated}. Indeed, these lemmas imply that $\cS\subseteq\firm$,
which allows us to obtain the inequality $\probla{\la\cS_n}\ge\probun{\un\cS_n}$ directly from~\refeq{la-un}.

The proof of Theorem \ref{thm:main} is complete.

\section{Canonical Cayley representations via \wl}\label{s:wl2}

This section is devoted to the proof of Theorem~\ref{thm:CCR}.
Before presenting the proof in Subsection~\ref{ss:CCR-proof},
we provide a formal description of the \wl algorithm in Subsection~\ref{ss:alg},
and collect relevant preliminary results in Subsections~\ref{ss:orbitals} and~\ref{ss:repr-of-firm}.

\subsection{Description of \wl}
\label{ss:alg}

For notational simplicity, an ordered pair $(a,b)$ will be denoted by $ab$.
Given a loopless digraph $X=(V,E)$ as an input, \wl iteratively computes a sequence of colorings $c^i_X$
of the Cartesian square $V\times V$. The initial coloring is defined by
$c^0_X(uv)=(\mathit{type}(uv),\mathit{type}(vu))$ where $\mathit{type}(uv)$
takes on one of three values according to the type of an ordered vertex pair $uv$,
namely \textit{arc} if $uv\in E$, \textit{nonarc}
if $uv\notin E$ and $u\ne v$, and \emph{loop} if $u=v$. The
coloring is modified iteratively as follows:
$$
  c^{i+1}_X(uv)=\Mset{\of{c^{i}_X(uw), c^{i}_X(wv)}}_{w\in V}.
$$
In words, the new color of a pair $uv$ is a ``superposition'' of all old color pairs observable
along the extensions of $uv$ to a triple $uwv$.
Denote the partition of $V\times V$ into the color classes of $c^i_X$ by $\cC^i_X$.
A simple inductive argument shows that 
$$
c^{i+1}_X(uv)\ne c^{i+1}_X(u'v')\text{ whenever }c^{i}_X(uv)\ne c^{i}_X(u'v')
$$
which means that $\cC^{i+1}_X$ is finer than or equal to $\cC^i_X$.
It follows that the partition stabilizes starting from some step
$t=t(X)$, that is, $\cC^{t}_X=\cC^{t-1}_X$, which implies
that $\cC^i_X=\cC^{t-1}_X$ for all $i\ge t$. Note that $t\le|V|^2$. As soon as
the stabilization is reached, the algorithm terminates and outputs the coloring~$c^t_X$.

An easy induction on $i$ shows that, if $\phi$ is an isomorphism from $X$ to $Y$, then
\begin{equation}
  \label{eq:phi}
c^i_X(uv)=c^i_Y(\phi(u)\phi(v)).  
\end{equation}

Note that the length of $c^i_X$-colors (in any natural encoding) grows exponentially with $i$ increasing.
Similarly to CR, the exponential blow-up is remedied by renaming the colors after each step.
Finally, note that \wl can be implemented in time $O(n^3\log n)$; see~\cite{ImmermanL90,ImmermanS19}.

Let $\wll X=\cC^{t}_X$ denote the stabilized partition of $V^2$.
It can be noticed that $\wll X$ forms a \emph{coherent configuration} \cite{CP2019},
but we will not use this fact directly. However, we will need another
source of coherent configurations, which we introduce in the next subsection.

\subsection{Orbitals of a permutation group}\label{ss:orbitals}

Let $G\le\sym(V)$ be a group of permutations of a set $V$.
A natural action of $G$ on $V^2$ is defined by $\alpha(x,y)=(\alpha(x),\alpha(y))$
for $x,y\in V$ and $\alpha\in G$. An orbit of this action is called \emph{orbital}.
We denote the partition of $V^2$ into the orbitals of $G$ by~$\orbb G$.

Equality \refeq{phi} readily implies that the partition $\wll X$ is refined by
the partition $\orbb{\aut(X)}$. Similarly to the concept of a Schurian coherent
configuration \cite{CP2019}, we call a digraph $X$ \emph{Schurian} if $\wll X=\orbb{\aut(X)}$.

For each $a\in\bZ_n$, a digraph $X=\cay(\bZ_n,S)$ has an automorphism $\sigma_a$
defined by $\sigma_a(x)=x+a$. These automorphisms form a subgroup of $\aut(X)$.
If this subgroup is normal, the circulant $X$ is called \emph{normal} (see \cite[Chapter 8.1]{DobsonMM22},
where this concept is discussed in the more general setting of Cayley graphs for any groups).
Note that if $X$ is firm---according to the definition given in Section \ref{ss:firm}, then $X$ is normal.
This is true both for digraphs and graphs (in the case of graphs, recall that every subgroup of index 2 is normal).

\begin{proposition}[{\cite[Theorem 6.1]{EvdokimovP03}}]\label{prop:EP}
  Every normal circulant digraph is Schurian.
\end{proposition}

Extending the definition of a firm Cayley (di)graph in Section \ref{ss:firm},
we also call an arbitrary labeled circulant (di)graph \emph{firm}
if it is isomorphic to a firm Cayley (di)graph $\cay(\bZ_n,S)$ or, equivalently,
if $\aut(X)$ is as small as possible, i.e., $\aut(X)$ isomorphic to the
cyclic group in the directed case and to the dihedral group in the undirected case.

The Schurity property of a digraph $X$ is beneficial because it
enables an efficient computation of the partition $\orbb{\aut(X)}$
just by running \wl on $X$. Moreover, if $X$ is a firm circulant (di)graph,
then the knowledge of $\orbb{\aut(X)}$ allows us to determine the automorphism
group $\aut(X)$ as a permutation group. We now state this fact formally
in the form that will be useful in the next subsection.

Let $C_n$ denote a cyclic permutation group on the $n$-element set $V$
that acts on $V$ transitively or, equivalently, contains a cycle of length $n$.
Assume that $n\ge3$ and note that $C_n$ has a unique extention to a dihedral group (of permutations of $V$).
We denote this dihedral permutation group by $D_{2n}$.
Note that all $\phi(n)$ elements of order $n$ in $C_n$ are cycles of length $n$.
The same holds true for $D_{2n}$. Indeed, every element in $D_{2n}\setminus C_n$
has degree~2. We now observe that, given the partitions $\orbb{C_n}$ and $\orbb{D_{2n}}$,
the elements of order $n$ in $C_n$ and $D_{2n}$ can be explicitly constructed
as permutations of the set~$V$.

\begin{lemma}\label{lem:generators}\hfill
  \begin{enumerate}[\bf 1.]
  \item
    The $\phi(n)$ generators of $C_n$ are uniquely determined by $\orbb{C_n}$
    and can be constructed from $\orbb{C_n}$ in time~$O(n^2)$.
  \item
    The $\phi(n)$ elements of order $n$ of $D_{2n}$ are uniquely determined by $\orbb{D_{2n}}$
    and can be constructed from $\orbb{D_{2n}}$ in time~$O(n^2)$.    
  \end{enumerate}
\end{lemma}

\begin{proof}
1.
Assume for a while that $V=\bZ_n$ and that $C_n$ consists of the permutations $\sigma_a$,
$a\in\bZ_n$, defined (like above) by $\sigma_a(x)=x+a$. Note that
$\sigma_a$ is a generator of $C_n$ if and only if $a\in\bZ_n^\times$.
Each orbital of $C_n$ will be regarded as a digraph. As easily seen, two pairs $(x,y)$ and $(x',y')$ are
in the same orbital exactly when $y-x=y'-x'$. It follows that the orbitals
of $C_n$ are exactly the digraphs $\cay(\bZ_n,\{a\})$ for $a\in\bZ_n$.
Note that $\cay(\bZ_n,\{a\})$ has $\gcd(a,n)$ connected components, and each of them
is isomorphic to the directed cycle of length $n/\gcd(a,n)$. This implies that, in general,
the generating elements of $C_n$ can be identified by finding all $\phi(n)$
orbitals of $C_n$ that, viewed as digraphs, are directed cycles of length $n$.
For each such cycle, one then forms a cyclic permutation of $V$ along the cycle.

2.
Note that $\orbb{C_n}$ contains, along with each orbital $C=\cay(\bZ_n,\{a\})$,
its transpose $C'=\cay(\bZ_n,\{-a\})$. As easily seen, the orbitals of $D_{2n}$ are
exactly the symmetric closures $C\cup C'=\cay(\bZ_n,\{-a,a\})$ of the orbitals of $C_n$.
Therefore, the elements of order $n$ of $D_{2n}$ are identified by finding all $\phi(n)/2$
orbitals of $D_{2n}$ that are (undirected) cycles of length $n$
and by forming, along each of these cycles, two cyclic permutations of~$V$ in both directions.
\end{proof}

\subsection{Cayley representations of firm circulants}\label{ss:repr-of-firm}

Recall that a Cayley representation of a labeled circulant $X$ on $n$ vertices
is a map $\lambda\function{V(X)}{\bZ_n}$ such that $X^\lambda$ is a Cayley graph,
i.e., $X^\lambda=\cay(\bZ_n,S)$ for some $S$. We call two Cayley representation $\lambda$ and $\lambda'$
of $X$ \emph{equivalent} if $X^\lambda=X^{\lambda'}$. If $X^\lambda=\cay(\bZ_n,S)$, then
$\lambda$ and $\lambda'$ are equivalent if and only if $\lambda'=\alpha\lambda$
for some automorphism $\alpha$ of $\cay(\bZ_n,S)$. Thus, if $X$ is a firm digraph, then $\lambda$ has exactly
$n$ equivalent Cayley representations, and if $X$ is a firm graph, then $\lambda$ has exactly
$2n$ equivalent representations. It is useful to notice the following simple fact.

\begin{lemma}\label{lem:repr-auto}\hfill
  \begin{enumerate}[\bf 1.]
  \item
    Let $X$ be a firm circulant digraph.
    There is a one-to-one correspondence between the equivalence classes of Cayley representations of $X$
    and the generators of~$\aut(X)$.
  \item
   Let $X$ be a firm circulant graph on $n$ vertices.
    There is a one-to-one correspondence between the equivalence classes of Cayley representations of $X$
    and the pairs of mutually inverse cycles of length $n$ in~$\aut(X)$.
  \item
    In both cases, if a cycle of length $n$ in $\aut(X)$ is given, a Cayley representation of $X$
    from the corresponding equivalence class is constructible in linear time.
  \end{enumerate}
\end{lemma}

\begin{proof}
Let $\lambda$ be a Cayley representation of $X$. The cycle $(\lambda^{-1}(0),\lambda^{-1}(1),\ldots,\lambda^{-1}(n-1))$
is a generator of $\aut(X)$. Every Cayley representation $\lambda'$ equivalent to $\lambda$
yields the same generator. Conversely, every cycle of length $n$ in $\aut(X)$ determines,
in an explicit way, $n$ equivalent Cayley representations of~$X$.

The case of graphs is similar with the only difference that a cycle $(v_0,v_1,\ldots,v_{n-1})$
and its inverse $(v_{n-1},\ldots,v_1,v_0)$ yield equivalent Cayley representations of~$X$.
\end{proof}

Lemma \ref{lem:repr-auto} implies that every firm circulant digraph $X$ on $n$ vertices has,
up to equivalence, exactly $\phi(n)$ Cayley representations. In the case that $X$ is a graph,
there are, up to equivalence, exactly $\phi(n)/2$ Cayley representations.

\subsection{Proof of Theorem \ref{thm:CCR}}\label{ss:CCR-proof}

We design an algorithm such that the following three conditions are fulfilled for a certain
class of (di)graphs~$\cC$:
\begin{itemize}
\item
  the algorithm computes a canonical Cayley representation for all inputs in~$\cC$;
\item
  the algorithm gives up on all inputs not in~$\cC$;
\item
  $\cC$ contains all firm circulant (di)graphs.
\end{itemize}
The last condition ensures the success probability bound stated in Theorem \ref{thm:CCR}.
Indeed, a random Cayley (di)graph $\cay(\bZ_n,S)$ is firm with high probability by Proposition \ref{prop:firm},
and this remains true for a random labeled circulant by the transition lemmas obtained in
Section \ref{s:main}, i.e., by Lemmas \ref{lem:u} and~\ref{lem:l}.

Since firm circulant (di)graphs are normal, Proposition \ref{prop:EP} shows that the orbital
partition $\orbb{\aut(X)}$ for a firm circulant $X$ can be computed just by running \wl on $X$.
According to \cite{ImmermanL90,ImmermanS19}, this takes time $O(n^3\log n)$, where $n$
is the number of vertices in $X$. Given $\orbb{\aut(X)}$, one can easily determine
all cycles of length $n$ in the permutation group $\aut(X)$; see Lemma \ref{lem:generators}.
By Lemma \ref{lem:repr-auto}, these can be used to efficiently construct all Cayley representations
of $X$. Note that it is enough to have one representation from each equivalence class.
Summing up, we come to the following procedure.

\medskip

\noindent\textsc{Canonical Cayley representation algorithm}

\smallskip

\noindent\textsc{Input:} a (di)graph $X$.
\begin{enumerate}
\item
  Run \wl on $X$.
\item
  Check whether the partition $\wll{X}$ contains a part that is isomorphic to a cycle (di)graph of length $n$.
If not, terminate. Otherwise, let $C$ be the part of $\wll{X}$ of this kind that has the lexicographically
smallest \wl-color.
\item
  Choose an arbitrary vertex $x_0$ and enumerate the vertices of $X$ along $C$ starting from $x_0$
  (in the case of graphs, choose any of the two directions in $C$). Let $x_0,x_1,\ldots,x_{n-1}$
  be the obtained enumeration.
\item
  Check whether the cyclic permutation $(x_0x_1\ldots x_{n-1})$ is an automorphism of $X$.
  If not, then give up.
\item
  Output the labeling $\lambda_X\function{V(X)}{\bZ_n}$ where $\lambda_X(x_i)=i$.
\end{enumerate}

The proof of Theorem~\ref{thm:CCR} is complete.

\begin{remark}
  Proposition \ref{prop:EP} can also be used to show that \wl distinguishes
  a firm circulant (di)graph $X$ from any other non-isomorphic (di)graph $Y$
  in the sense that the color palettes produced by \wl on $X$ and $Y$
  are different, i.e., $\Mset{c^{t(X)}_X(uv)}{uv\in V(X)^2}\ne \Mset{c^{t(Y)}_Y(uv)}{uv\in V(Y)^2}$.
  The class of the firm circulant digraphs contains all circulant digraphs with simple spectrum \cite[Theorem 3]{ElspasT70}.
  For this smaller class of circulant digraphs,
  the identifiability by \wl follows also from any of the results stated in \cite[Corollary 4.5]{Friedland89} or
  \cite[Corollary of Theorem 1]{EP99} (the latter results is stronger than the former in view of \cite[Theorem 3.3.19]{CP2019}).
\end{remark}


\end{document}